\documentclass[a4paper,11pt]{article}
\pdfoutput=1

\usepackage{jheppub}

\usepackage[T1]{fontenc}
\usepackage{amsmath,amssymb,amsthm,bm}
\usepackage{graphicx}
\usepackage{microtype}
\usepackage{booktabs}
\usepackage{makecell}
\usepackage{tikz}
\usetikzlibrary{arrows.meta}

\newcommand{\dd}{\mathrm{d}}
\newcommand{\cV}{\mathcal V}
\newcommand{\cX}{\mathcal X}
\newcommand{\cW}{\mathfrak W}
\newcommand{\Tr}{\operatorname{Tr}}
\newcommand{\Lk}{\operatorname{Lk}}
\newcommand{\cA}{\mathcal A}
\newcommand{\cU}{\mathcal U}
\newcommand{\cG}{\mathcal G}
\newcommand{\EPH}{E_{\rm P}^{\rm(h)}}
\newcommand{\cP}{\mathfrak P}
\newcommand{\fmark}{\mathfrak f}
\newcommand{\arcsinh}{\operatorname{arcsinh}}
\newcommand{\arccosh}{\operatorname{arccosh}}
\newcommand{\cC}{\mathcal C}
\newcommand{\Alg}{\mathfrak A}
\newcommand{\cQ}{\mathcal Q}
\newcommand{\cZ}{\mathcal Z}

\theoremstyle{plain}
\newtheorem{theorem}{Theorem}[section]
\newtheorem{corollary}[theorem]{Corollary}
\newtheorem{proposition}[theorem]{Proposition}
\theoremstyle{definition}
\newtheorem{definition}[theorem]{Definition}
\theoremstyle{remark}
\newtheorem{remark}[theorem]{Remark}

\title{\boldmath Holonomy Separation and Anyonic Hair in AdS$_3$}

\author{Kumar Ghosh}
\emailAdd{jb.ghosh@outlook.com}
\affiliation{E.ON Digital Technology, Laatzener Str.\ 1, 30539 Hannover, Germany}


\abstract{
Three-dimensional Einstein gravity and Chern-Simons-Higgs matter both admit
Chern-Simons descriptions, but their holonomies encode different physics. We
prove that for any Einstein-Chern-Simons-Higgs solution with
charge-conjugation-invariant boundary data, charge conjugation reverses the
vortex winding $n\to-n$ while leaving the metric and first-order gravitational
connections pointwise unchanged. Every closed gravitational holonomy
character, every matched Wilson amplitude, the BTZ charges, the complete
classical Ba\~nados data and the leading geometric mixed-state hierarchy are
therefore exactly even in $n$. The result extends beyond classical geometry.
On a fixed marked Euclidean replica filling with a charge-conjugation-covariant
measure, gauge fixing and regulator, a change of variables gives
$Z_{q,\fmark}[n;\mathcal J_+]=Z_{q,\fmark}[-n;\mathcal J_+]$ for every replica
number and every charge-conjugation-even source, so the uncharged R\'enyi
hierarchy, the Faulkner--Lewkowycz--Maldacena one-loop correction, the
renormalized generalized-entropy functional and the set of quantum extremal
surfaces coincide for the two orientations. The same change of variables fixes
the complementary sector exactly. With an Aharonov-Bohm angle $\mu$ conjugate
to the entangling-region charge, $Z_{q,\fmark}[n;\mu]=Z_{q,\fmark}[-n;-\mu]$,
so the symmetry-resolved entanglement spectrum obeys
$\cZ_q(\cQ;n)=\cZ_q(-\cQ;-n)$ and every odd charged cumulant reverses with the
winding. The orientation is thus stored entirely in the charge asymmetry of
the entanglement spectrum, and is read at topological order by a linked
charged matter line with contrast
$\cX_p^{(\fmark)}(n)=\exp[-2i\nu_{\fmark}\kappa\Phi_p\Phi_n]$, which aliases
the two orientations exactly when the vortex sector is self-conjugate. The
pair is a bulk logical qubit on which the entire neutral
gravitational-replica algebra acts as a multiple of the identity. Exact BTZ
cross sections show that the purification plateau and the entanglement-wedge
disconnection transition are orientation blind, the latter closing the
cross-section-supported readout with the cutoff-independent jump
$(c/3)\log(1+\sqrt2)$ and Markov gap $(c/3)\log(3+2\sqrt2)$. Numerical
Bogomolny solutions verify odd flux and even energy and spin. Using the same
profiles only as linearized charge-matching data, we also obtain the
dimensionless asymptotic shifts $\delta\widehat M=16\pi Gv^2|n|$ and
$|\delta\widehat j|=16\pi Gv^2n^2/(mL)$. Conditional on the existence of a
stationary localized Einstein-CSH branch and on the rotating-BTZ
cross-section continuation stated explicitly below, these shifts define a
finite continuous winding envelope. This perturbative backreaction estimate
is kept separate from the exact charge-conjugation theorem.
}

\begin{document}
\maketitle
\flushbottom

\section{Introduction}
\label{sec:intro}

A central limitation of semiclassical gravity is that physically distinct
matter configurations can generate identical gravitational data. In $2+1$
dimensions this limitation can be sharpened into an exact statement. The Weyl
tensor vanishes identically, so local curvature is fixed algebraically by the
stress tensor and there are no propagating Einstein-gravity modes. Nontrivial
classical gravitational information resides instead in global holonomies
around noncontractible curves and in boundary gravitons generated by
asymptotic symmetries~\cite{Carlip2005Review,Carlip2023Review}. Vacuum
AdS$_3$ gravity can be written as an
$\mathrm{SL}(2,\mathbb R)_+\times\mathrm{SL}(2,\mathbb R)_-$ Chern-Simons
theory~\cite{AchucarroTownsend1986,Witten1988Gravity,Carlip1998Book},
Brown-Henneaux boundary conditions produce two Virasoro
algebras~\cite{BrownHenneaux1986}, and stationary BTZ geometries are
characterized by the conjugacy classes of their two angular-cycle
holonomies~\cite{BTZ1992,BHTZ1993,Carlip1995BTZ}. Because the classical data
separate cleanly into local curvature, global holonomy and asymptotic
boundary modes, one can hope to exhaust them, and so to decide what geometry
cannot see.

Self-dual Chern-Simons-Higgs (CSH) vortices supply a second and physically
distinct Chern-Simons structure. Their Abelian Gauss law binds electric
charge to magnetic flux, and their topological sectors carry fractional spin
and nontrivial mutual
monodromy~\cite{HongKimPac1990,JackiwWeinberg1990,JackiwLeeWeinberg1990}. The
two Chern-Simons formulations must not be conflated. The gravitational
connections encode geometry, parallel transport and asymptotic charges, while
the matter connection encodes charge-flux attachment and anyonic transport.
Charge conjugation acts trivially on the first and reverses the second, and
it is this separation that the present article makes precise.

For any Lorentzian Einstein-CSH solution obeying charge-conjugation-invariant
boundary conditions, the map
\begin{equation}
(g_{\mu\nu},\phi,A_\mu)\longmapsto(g_{\mu\nu},\phi^*,-A_\mu)
\label{eq:introCmap}
\end{equation}
constructs a winding $-n$ solution with the same metric from a winding $n$
solution. No weak-backreaction expansion, radial ansatz, stationarity
assumption or uniqueness theorem is required. In first-order variables the
triad and spin connection may be chosen identically, so all closed
gravitational holonomy characters coincide, and Wilson lines or networks
coincide whenever their contours, representations, endpoint states, junction
intertwiners and boundary frames are matched. When the solution is stationary
and asymptotically BTZ, its holonomy classes, mass, angular momentum, horizon
data and Brown-Henneaux zero modes coincide, and more generally, when the
matter falloff preserves Brown-Henneaux boundary conditions, the complete
classical Ba\~nados functions and every Virasoro mode coincide. The vortex
flux nevertheless reverses, and the matter mutual monodromy is inverted. The
pair is therefore indistinguishable in every classical Einstein-gravity
channel at fixed dressing, and distinguishable by a topological matter line.

The same symmetry controls the quantum replica sector, and it does so in two
complementary ways. For each replica number and each fixed marked filling,
charge conjugation is a bijection between the winding-$n$ and winding-$-n$
integration cycles. When the measure, gauge fixing and regulator are
charge-conjugation covariant, the corresponding replica partition functions
agree for every charge-conjugation-even source. The equality is therefore not
merely an equality of saddle metrics. It pairs the fluctuation operators,
their ghost sectors and their renormalized determinants, and it makes the
generalized entropy equal as a functional of every candidate surface, so that
neither the Faulkner--Lewkowycz--Maldacena correction nor quantum
extremization can split the two
orientations~\cite{FaulknerLewkowyczMaldacena2013,EngelhardtWall2015}. The
same change of variables determines the charged sector exactly rather than
leaving it open. Inserting an Aharonov-Bohm angle $\mu$ conjugate to the
charge of the entangling region gives
$Z_{q,\fmark}[n;\mu]=Z_{q,\fmark}[-n;-\mu]$, whose even part in $\mu$ is even
in $n$ and whose odd part is odd in $n$. Summing over the charge sectors
returns the blind total entropy, so the orientation is carried by the charge
asymmetry of the entanglement spectrum and by nothing else. This locates the
distinguishing observable before any particular probe is constructed.

What makes this worth proving is not that it yields one further invariant. It
is that an entire observable algebra goes blind to a single bit of bulk
information, and that quantum gravity forbids such a bit from being
unreadable in principle. Harlow and Ooguri show that a bulk symmetry acting
faithfully on matter while acting trivially on all gravitationally dressed
data cannot be an exact global symmetry of a holographic
theory~\cite{HarlowOoguri2019,HarlowOoguri2021}. The theorem above constructs
precisely that configuration, explicitly and nonperturbatively in Newton's
constant, so their analysis applies and tells us what must complete it. The
matter sector has to supply a charged operator that reads the bit. It does,
in the form of a topological line, and that line is accordingly not an
optional probe added for convenience but an object the consistency of the
theory requires.

Holographic purification is the natural setting in which to watch both halves
at once. The entanglement of purification minimizes the entropy over all
purifications of a mixed state, and Nguyen, Devakul, Halbasch, Zaletel and
Swingle, simultaneously with Takayanagi and Umemoto, proposed that its
leading holographic value is the minimal cross section of the entanglement
wedge, $\EPH(A{:}B)=E_W(A{:}B)/(4G)$~\cite{NguyenEtAl2018,TakayanagiUmemoto2018}.
The canonical purification underlying reflected entropy gives the
complementary relation $S_R=2E_W/(4G)=2\EPH$ at leading
order~\cite{DuttaFaulkner2021}, so a single bulk cross section controls two
distinct mixed-state correlation measures. That same cross section also
decides whether a closed matter contour can be supported at all, and this is
what turns geometry into a gate.

For spin-$2$ gravity, Ref.~\cite{AmmonCastroIqbal2013} represents the RT
geodesic of a boundary interval by an open gravitational Wilson line in the
appropriate infinite-dimensional highest-weight representation, with endpoint
conditions $U_i=U_f=\mathbf 1$ and the replica-limit Casimir assignment
$\sqrt{2c_2}\to c/6$. Since the gravitational connections are identical for
$n$ and $-n$, these identically dressed RT amplitudes coincide, and equality
of the EWCS, of $\EPH$ and of $S_R$ follows directly from equality of the two
metrics. We then add a matter-sector vortex line to the Euclidean
purification replica and normalize away all one-line factors. The replica
saddle must be specified as a marked filling $\fmark$ of the boundary torus,
including its contractible cycle, framing, gluing prescription and common
edge-sector data~\cite{Witten1989,ElitzurEtAl1989,Witten1992Sewing}. If
$\nu_{\fmark}$ is the signed linking number in that fixed filling, the
surviving phase is the Magnus-corrected mutual monodromy, and the
vortex-antivortex contrast approaches
\begin{equation}
\cX_p^{(\fmark)}(n)=
\exp[-2i\nu_{\fmark}\kappa\Phi_p\Phi_n].
\label{eq:introcontrast}
\end{equation}

The analytic BTZ formulas of Ref.~\cite{NguyenEtAl2018} make the gate
quantitative. The geometric purification cost can develop plateaus and can
jump to zero when the entanglement wedge disconnects, yet it is identical for
the two vortex orientations. Two exact properties of those cross sections
follow in a few lines of algebra and are recorded in
Sec.~\ref{sec:purificationgate}. A single inequality
$r_+/L\geq(2/\pi)\arcsinh 1$ governs both the existence of the one-sided
plateau and the existence of the two-sided transition, and the jump of
$\EPH$ at that transition is the cutoff-independent value
$(L/2G)\arcsinh 1=(c/3)\log(1+\sqrt2)$. The latter number is fixed by the
central charge alone, and it carries no information whatever about the sector
whose readout it gates. The cross-section-supported matter phase, by
contrast, stays quantized throughout a connected linked saddle and becomes
unavailable only when that channel closes. The same conclusion holds for the
first quantum correction: any bulk mutual information that survives the
classical disconnection at order $G^0$ is charge-conjugation even and hence
again common to the two orientations. Holographic purification therefore acts
as a gate rather than as an orientation detector, at leading order and at one
loop alike.

The results themselves are quickly stated. We prove the exact
charge-conjugation pairing for any admissible Einstein-CSH solution with
invariant boundary sources, and show that it leaves the complete classical
gravitational fields unchanged. Specializing to Brown-Henneaux asymptotics
gives equality of local curvature, of all matched gravitational Wilson data,
of the full classical Ba\~nados functions and every Virasoro mode, and, on
the stationary BTZ branch, of the mass, angular momentum, horizon radii and
semiclassical Cardy entropy. The leading geometric purification hierarchy,
comprising $\EPH$, the reflected entropy, the mutual information and the
Markov gap, is blind in the same way. At fixed marked filling this blindness
extends to every charge-conjugation-even replica partition function and, in
particular, to the FLM one-loop determinant, the renormalized generalized
entropy and the quantum-extremal-surface set. The charge-conjugate pair then
spans a two-dimensional code subspace on which the neutral
classical-and-replica algebra acts as a multiple of the identity, so the
relative orientation is a bulk logical qubit in the sense of holographic
quantum error
correction~\cite{AlmheiriDongHarlow2015,DongHarlowWall2016,JahnEisert2021}.
We then determine the complementary sector in closed form. The charged
replica identity $Z_{q,\fmark}[n;\mu]=Z_{q,\fmark}[-n;-\mu]$ implies
$\cZ_q(\cQ;n)=\cZ_q(-\cQ;-n)$ for the symmetry-resolved partition functions
and entropies, and makes every odd charged cumulant exactly odd in $n$, with
leading topological value $\kappa\Phi_n$ for a wedge containing the vortex
core. We construct the operator that realizes this odd datum at topological
order as a Euclidean line insertion on a fixed marked replica filling, derive
its universal contrast together with its framing, edge-sector and modular
qualifications, and show that it fails to resolve the two orientations if and
only if the vortex sector is self-conjugate, which is to say if and only if
there is no doublet to resolve. We then ask what a weakly backreacted exterior can retain if a regular
stationary vortex-dressed branch exists. Using the flat-space BPS core only
as linearized charge-matching data, the dimensionless BTZ shifts are
$\delta\widehat M=\gamma|n|$ and
$|\delta\widehat j|=\gamma n^2/(mL)$ for a nonrotating seed, with
$\gamma=16\pi Gv^2$. Combining these even shifts with the rotating-BTZ
matching diagnostic gives a finite continuous candidate envelope for $|n|$.
This envelope is not itself a spectrum: physical candidates must additionally
have integer winding, a small matter tail outside the horizon, and a regular
stationary Einstein-CSH completion. Flat-space Bogomolny boundary-value
solutions verify the odd parity of magnetic flux and the even parity of
energy and spin for $|n|=1,2,3$, with all three sum-rule errors below
$1.1\times10^{-11}$.

Five qualifications hold throughout and are stated here once rather than
repeated at each use. Throughout, $\EPH$ denotes the holographic
cross-section prescription, and we do not assert that the unrestricted
field-theoretic minimization over all purifications has been proved to equal
it. No standalone gauge-invariant Wilson-network representation of the EWCS
is assumed, since equality of the cross sections for the two orientations
follows from equality of the metrics. Every line-amplitude formula refers to
one fixed marked filling, with identical framing and edge-sector data in
numerator and denominator, and no modular sum over fillings is evaluated. The
quantum replica statements are termwise fixed-filling statements, and assume a
charge-conjugation-covariant operator-algebra or edge-mode prescription, gauge
fixing and regulator, together with a common continuation in the replica
number. The classical charge-conjugation and holonomy-separation theorems of
Sec.~\ref{sec:cthm} are unconditional, whereas the BTZ specialization
presupposes a regular positive-mass nonextremal Einstein-CSH branch that is
not constructed here, as Appendix~\ref{sec:scope} states in full. The
backreaction analysis of Secs.~\ref{sec:exterior} and~\ref{sec:band} is a
linearized charge-matching estimate: it additionally assumes that the BPS core
controls the integrated $O(G)$ charge shifts, that the matter tail outside the
horizon is small, and, for the winding envelope, that the rotating-BTZ
cross-section matching used in Appendix~\ref{sec:rotatingplateau} is the
relevant covariant continuation. None of these additional assumptions enters
the exact orientation-blindness theorem. Finally, the error-correction reading
of Sec.~\ref{sec:qec} interprets geometric facts established independently in
Sec.~\ref{sec:purificationgate} rather than replacing them.

Section~\ref{sec:setup} defines the CSH model and the flat-space Bogomolny
core data. Section~\ref{sec:cthm} states the hypotheses, proves the
charge-conjugation and holonomy-separation theorems, derives their local,
global and asymptotic consequences, and proves the fixed-filling replica
identity together with its FLM, quantum-extremal-surface and
symmetry-resolution corollaries. Section~\ref{sec:purificationgate} uses analytic one-sided and
two-sided BTZ cross sections to establish the purification gate, then develops
a conditional linearized charge-matching estimate and its continuous winding
envelope.
Section~\ref{sec:interferometer} constructs the channel-selective line
insertion and derives the vortex-antivortex contrast.
Section~\ref{sec:qec} recasts these results as statements about a bulk
logical qubit, fixes exactly when the doublet exists, and shows why the
matter line is required. Section~\ref{sec:discussion} discusses the resulting
gravitationally degenerate topological doublet.
Appendix~\ref{sec:gravcsapp} records first-order and BTZ conventions,
Appendix~\ref{sec:btzpurapp} collects the analytic purification benchmarks,
derives the universal nonrotating jump height and records the conditional
rotating-BTZ matching diagnostic, and Appendix~\ref{sec:scope} states precisely the scope of the
assumed asymptotically BTZ branch.

\section{Chern-Simons-Higgs vortex sector and flat-space Bogomolny data}
\label{sec:setup}

\subsection{Action and conventions}
\label{sec:action}

Consider Einstein gravity with negative cosmological constant coupled to a
relativistic CSH model,
\begin{align}
I&=I_{\rm EH}+I_{\rm CSH},\nonumber\\
I_{\rm EH}&=\frac{1}{16\pi G}\int \dd^3x\sqrt{-g}
\left(R+\frac{2}{L^2}\right),\nonumber\\
I_{\rm CSH}&=\int \dd^3x\sqrt{-g}
\left[(D_\mu\phi)^*D^\mu\phi-V(|\phi|)\right]
+\frac{\kappa}{4}\int \dd^3x\,
\epsilon^{\mu\nu\rho}A_\mu F_{\nu\rho},
\label{eq:action}
\end{align}
where $D_\mu=\nabla_\mu-ieA_\mu$ and
\begin{equation}
V(|\phi|)=\frac{e^4}{\kappa^2}|\phi|^2\left(|\phi|^2-v^2\right)^2.
\label{eq:potential}
\end{equation}

Equation~\eqref{eq:potential} is the self-dual potential of the flat-space
relativistic CSH theory. Coupling this fixed potential to Einstein gravity
preserves the exact charge-conjugation symmetry used below, but does not by
itself establish a fully backreacted Bogomolny system. In related gravitating
CSH models a first-order reduction requires a gravitationally modified
eighth-order potential~\cite{London1995,Clement1996}. The numerical profiles
of Sec.~\ref{sec:profiles} are accordingly controlled flat-space BPS core
data, and are not claimed to solve the coupled Einstein-CSH boundary-value
problem. The charge-conjugation theorem does not require self-duality and
holds for any potential depending only on $|\phi|$, including such
gravitational completions.

We work with signature $(+,-,-)$ for the flat-space vortex calculation and
$\epsilon^{012}=+1$ for the antisymmetric density. Reversing the orientation
convention for the Chern-Simons term reverses all displayed flux and spin
phases simultaneously, and leaves the observable distinction between $n$ and
$-n$ unchanged. The Chern-Simons term is metric independent, so it does not
contribute to the matter stress tensor.

\subsection{Bogomolny structure and self-dual equations}
\label{sec:bps}

Variation of the action with respect to $A_0$ gives the Chern-Simons Gauss
constraint, which with the orientation of Sec.~\ref{sec:action} integrates to
\begin{equation}
Q=\kappa\Phi,\qquad
\Phi=\int \dd^2x\,B,
\label{eq:gauss}
\end{equation}
so any excitation carrying magnetic flux automatically carries proportional
electric charge. At the sixth-order potential in Eq.~\eqref{eq:potential} the
static energy rearranges into nonnegative squares plus a total flux
term~\cite{HongKimPac1990,JackiwWeinberg1990,JackiwLeeWeinberg1990}. For
positive flux the first-order Bogomolny equations are
\begin{align}
(D_1+iD_2)\phi&=0,\nonumber\\
B&=\frac{2e^3}{\kappa^2}|\phi|^2\left(v^2-|\phi|^2\right),\nonumber\\
A_0&=-\frac{e}{\kappa}\left(v^2-|\phi|^2\right),
\label{eq:BPS}
\end{align}
with simultaneous sign reversal for negative flux. Solutions of
Eq.~\eqref{eq:BPS} saturate the Bogomolny bound
\begin{equation}
E\ge ev^2|\Phi|.
\label{eq:bound}
\end{equation}
For a topological vortex, finite energy requires $|\phi|\to v$ and
$D_i\phi\to0$ at spatial infinity, so the flux is quantized,
\begin{equation}
\Phi_n=\frac{2\pi n}{e},\qquad
E_n=2\pi v^2|n|.
\label{eq:quantized}
\end{equation}
The same solutions carry intrinsic angular momentum
\begin{equation}
s_n=J_n=-\frac{\kappa\Phi_n^2}{4\pi},
\label{eq:spin}
\end{equation}
which is even under $n\to-n$.

A normalization that we use repeatedly below follows directly from the third
line of Eq.~\eqref{eq:BPS}. On the self-dual branch the electric contribution
to the energy density coincides with the potential,
\begin{equation}
|D_0\phi|^2=e^2A_0^2|\phi|^2
=\frac{e^4}{\kappa^2}|\phi|^2\left(v^2-|\phi|^2\right)^2
=V(|\phi|),
\label{eq:electricequalspotential}
\end{equation}
so that the two terms combine into $2V$ in the static energy density. This
fixes the coefficient of the potential term in Eq.~\eqref{eq:profiles}.

\subsection{Vortex charges and topological data}
\label{sec:charges}

Collecting the quantized flux, the Chern-Simons attached charge, the BPS
energy and the intrinsic spin gives the topological data of the sector,
\begin{align}
\Phi_n&=\frac{2\pi n}{e}, & Q_n&=\kappa\Phi_n,\nonumber\\
E_n&=ev^2|\Phi_n|, & s_n&=-\frac{\kappa\Phi_n^2}{4\pi}.
\label{eq:charges}
\end{align}
Here $s_n$ is the intrinsic vortex spin in units with $\hbar=1$, and
Eqs.~\eqref{eq:charges} hold up to the common orientation convention for
$\kappa$ and $\Phi_n$. The flux and charge are odd under $n\to-n$, while the
energy and spin are even. This parity asymmetry is the single input required
for both the geometric-blindness theorem of Sec.~\ref{sec:cthm} and the
interferometric readout of Sec.~\ref{sec:interferometer}.

\subsection{Radial profiles and numerical validation}
\label{sec:profiles}

For a rotationally symmetric positive-winding solution take
\begin{equation}
\phi=vg(r)e^{in\varphi},
\qquad
A_\varphi=\frac{n-a(r)}{e},
\label{eq:ansatz}
\end{equation}
with boundary conditions
\begin{equation}
g(0)=0,\quad a(0)=n,
\qquad
g(\infty)=1,\quad a(\infty)=0.
\label{eq:bc}
\end{equation}
Defining $m=2e^2v^2/|\kappa|$ and $x=mr$, the BPS system in
Eq.~\eqref{eq:BPS} reduces to
\begin{equation}
\frac{\dd g}{\dd x}=\frac{ag}{x},\qquad
\frac{\dd a}{\dd x}=-\frac{x}{2}g^2\left(1-g^2\right).
\label{eq:radial}
\end{equation}
The local field densities on this ansatz are
\begin{align}
\frac{eB_n}{m^2}&=\operatorname{sgn}(n)\,\frac{g^2(1-g^2)}{2},\nonumber\\
\frac{T_{00}}{m^2v^2}&=2\left(\frac{ag}{x}\right)^2
+\frac{g^2(1-g^2)^2}{2},\nonumber\\
\frac{T_{0\varphi}}{mv^2}&=-\operatorname{sgn}(\kappa)\,a\,g^2\left(1-g^2\right),
\label{eq:profiles}
\end{align}
where the gradient term uses the first equation of Eq.~\eqref{eq:radial} and
the potential term carries the factor $2$ supplied by
Eq.~\eqref{eq:electricequalspotential}. The magnetic field density vanishes
at the vortex centre because $g(0)=0$, so $B_n$ forms a ring rather than
peaking at the scalar zero.

\begin{figure}[t]
\centering
\includegraphics[width=0.99\textwidth]{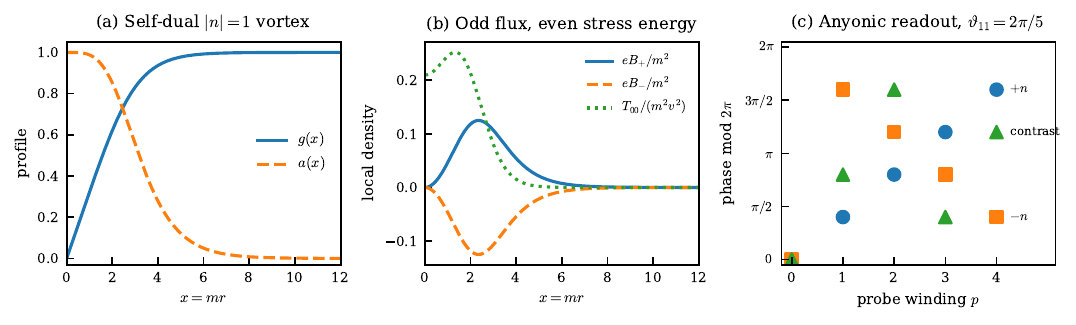}
\caption{Self-dual CSH vortex and its entanglement-interferometric readout.
(a) Numerical Bogomolny profiles for $|n|=1$, with $x=mr$ and
$m=2e^2v^2/|\kappa|$. (b) Charge conjugation reverses the magnetic field,
$B_{+}=-B_{-}$, but leaves the local energy density identical. The magnetic
field is ring shaped because it vanishes where the Higgs field vanishes.
(c) Representative topological-limit phases for the unit-vorticity
mutual-monodromy angle $\vartheta_{11}\equiv-\kappa(2\pi/e)^2=2\pi/5$ and
unit marked-filling linking $\nu_{\fmark}=1$. The contrast phase is twice the
single-orientation phase modulo $2\pi$.}
\label{fig:main}
\end{figure}

We solved Eq.~\eqref{eq:radial} as a boundary-value problem on
$x\in[10^{-5},35]$ using an adaptive collocation solver. Near the origin we
imposed the regular series
\begin{equation}
g(x)=C_nx^n+O(x^{n+2}),\quad
a(x)=n-\frac{C_n^2}{4(n+1)}x^{2n+2}+\cdots,
\label{eq:series}
\end{equation}
and determined $C_n$ by the condition $g(35)=1$. The integrals were
evaluated on a uniform grid. The three dimensionless sum rules are
\begin{align}
2\pi\int_0^\infty x\,\frac{eB}{m^2}\,\dd x&=2\pi n,\nonumber\\
2\pi\int_0^\infty x\,\frac{T_{00}}{m^2v^2}\,\dd x&=2\pi |n|,\nonumber\\
\int_0^\infty x\,\frac{|T_{0\varphi}|}{mv^2}\,\dd x&=n^2,
\label{eq:sumrules}
\end{align}
and they carry the flux, energy and magnitude of the spin kernel of a single
topological vortex. The first is odd in $n$, whereas the latter two are even.
The physical spin is recovered as
$J_n=-2\pi\operatorname{sgn}(\kappa)(v^2/m)n^2
=-\kappa\Phi_n^2/(4\pi)$, in agreement with Eq.~\eqref{eq:spin}.

The sum rules play different diagnostic roles, and we state this explicitly to
avoid overinterpretation. The flux and spin integrands are both exact
derivatives on any solution of Eq.~\eqref{eq:radial}. For the flux,
\begin{equation}
2\pi\int_0^\infty x\,\frac{g^2(1-g^2)}{2}\,\dd x
=-2\pi\int_0^\infty \frac{\dd a}{\dd x}\,\dd x
=2\pi\left[a(0)-a(\infty)\right]=2\pi n,
\label{eq:fluxfirstintegral}
\end{equation}
while the second line of Eq.~\eqref{eq:radial} gives
$x\,g^2(1-g^2)=-2\,\dd a/\dd x$ and hence
\begin{equation}
\int_0^\infty x\,a\,g^2(1-g^2)\,\dd x
=-\int_0^\infty\frac{\dd\left(a^2\right)}{\dd x}\,\dd x
=a(0)^2-a(\infty)^2=n^2,
\label{eq:spinfirstintegral}
\end{equation}
which reproduces the magnitude of $s_n\propto\Phi_n^2$ in
Eq.~\eqref{eq:charges} directly from the profile once the prefactor
$-2\pi\operatorname{sgn}(\kappa)v^2/m$ is restored. These two rules therefore test the accuracy with which the solver
realizes the boundary conditions in Eq.~\eqref{eq:bc} and do not independently
test the BPS structure. The energy sum rule is the nontrivial check, since
saturation of Eq.~\eqref{eq:bound} is a property of the self-dual branch and
not of a generic solution of the boundary-value problem.

Results are collected in Table~\ref{tab:validation} and the resulting
profiles are displayed in Fig.~\ref{fig:main}(a). For $n=1,2,3$ all three
integrals reproduce their exact values with relative errors below
$1.1\times10^{-11}$. The flux is odd under $n\to-n$, while every local source
entering the gravitational response is even. The ring-shaped field in
Fig.~\ref{fig:main}(b) is a distinctive CSH
feature~\cite{JackiwLeeWeinberg1990}.

One further property of the self-dual branch is used repeatedly below.
Linearizing Eq.~\eqref{eq:radial} about $g=1$ with $g=1-\eta$ gives
$\eta''+\eta'/x-\eta=0$, so $\eta\propto K_0(x)$ and the source falls off as
\begin{equation}
\frac{T_{00}}{m^2v^2}\ \propto\ \frac{e^{-2x}}{x},
\qquad x\to\infty .
\label{eq:exponentialfalloff}
\end{equation}
Fitting $\log[x\,T_{00}]$ on $12\leq x\leq26$ returns slopes
$-2.00075$, $-2.00073$ and $-2.00049$ for $n=1,2,3$, confirming
Eq.~\eqref{eq:exponentialfalloff}. The half-mass radii are
$mR_n=2.055$, $4.101$ and $6.114$, so the core grows linearly,
$R_n\simeq 2n/m$. Exponential localization is specific to the sixth-order
self-dual potential of Eq.~\eqref{eq:potential}, and
Sec.~\ref{sec:exterior} turns it into a statement about the exterior
geometry. The supplied scripts reproduce every table and figure in this
article from a clean working directory.

\begin{table}[t]
\centering
\caption{Numerical validation of the BPS solutions. The last three columns
show the relative errors in the flux, energy and spin sum rules of
Eq.~\eqref{eq:sumrules}.}
\label{tab:validation}
\begin{tabular}{cccccc}
\toprule
$n$ & $C_n$ & \makecell{maximum\\ODE residual} & flux error & energy error & spin error\\
\midrule
1 & $3.238614446\times10^{-1}$ & $1.52\times10^{-8}$ & $4.4\times10^{-13}$ & $1.0\times10^{-11}$ & $2.4\times10^{-13}$\\
2 & $3.89720102\times10^{-2}$ & $2.25\times10^{-8}$ & $3.0\times10^{-13}$ & $3.8\times10^{-14}$ & $1.3\times10^{-14}$\\
3 & $2.85136584\times10^{-3}$ & $2.98\times10^{-8}$ & $2.5\times10^{-12}$ & $2.3\times10^{-14}$ & $3.6\times10^{-14}$\\
\bottomrule
\end{tabular}
\end{table}

\section{Holonomy-separation theorem in Chern-Simons gravity}
\label{sec:cthm}

\subsection{Two Chern-Simons structures}
\label{sec:twocs}

Introduce a triad one-form $e^a=e^a{}_\mu\dd x^\mu$ and the dualized spin
connection
\begin{equation}
\omega^a=\frac{1}{2}\epsilon^{abc}\omega_{\mu bc}\dd x^\mu,
\qquad
g_{\mu\nu}=\eta_{ab}e^a{}_\mu e^b{}_\nu .
\label{eq:firstorder}
\end{equation}
For an invertible triad and AdS radius $L$, define the two gravitational
connections
\begin{equation}
\cA^{(\pm)a}=\omega^a\pm\frac{1}{L}e^a .
\label{eq:gravconnections}
\end{equation}
Up to the standard boundary terms and with the standard invariant trace
normalization, the Einstein-Hilbert action is
\begin{equation}
I_{\rm EH}=I_{\rm CS}[\cA^{(+)}]-I_{\rm CS}[\cA^{(-)}],
\qquad
k_{\rm grav}=\frac{L}{4G},
\label{eq:EHCS}
\end{equation}
where
\begin{equation}
I_{\rm CS}[\cA]=\frac{k_{\rm grav}}{4\pi}
\int\Tr\left(\cA\wedge\dd\cA+\frac{2}{3}\cA\wedge\cA\wedge\cA\right).
\label{eq:gravCS}
\end{equation}
In vacuum the field equations are the flatness conditions
$F[\cA^{(\pm)}]=0$, so the gauge-invariant bulk data are encoded by
holonomies of the two
connections~\cite{AchucarroTownsend1986,Witten1988Gravity,Carlip2005Review}.
In the presence of the vortex core the gravitational connections are sourced
and need not be flat, but the first-order formulation remains valid. The
Abelian gauge field $A_\mu$ in Eq.~\eqref{eq:action} is a distinct
Chern-Simons connection. It carries the vortex charge-flux data and is not to
be identified with either gravitational connection in
Eq.~\eqref{eq:gravconnections}.

\subsection{Hypotheses}
\label{sec:hypotheses}

We isolate the assumptions so that the scope of the theorem is unambiguous.

\begin{definition}[Admissible branch]
\label{def:admissible}
A configuration $(g_{\mu\nu},\phi,A_\mu;\mathcal B)_n$ is an
\emph{admissible branch} of winding $n$ if
\begin{itemize}
\item[\rm(H1)] it solves the Lorentzian field equations of
Eq.~\eqref{eq:action} on a fixed manifold, with $V$ a function of $|\phi|$
only;
\item[\rm(H2)] the triad is invertible, so that the first-order variables in
Eq.~\eqref{eq:firstorder} are well defined;
\item[\rm(H3)] the complete gravitational and matter boundary data
$\mathcal B$ are invariant under the charge-conjugation map ${\sf C}$ of
Eq.~\eqref{eq:C}.
\end{itemize}
\end{definition}

\begin{definition}[Matched dressing data]
\label{def:dressing}
Two Wilson observables built on charge-conjugate configurations are said to
carry \emph{matched dressing data} if their contour, group representation,
initial and final endpoint states, boundary Lorentz frame and, for a network,
junction intertwiners, are chosen identically.
\end{definition}

Hypothesis (H3) is the operative restriction. For Dirichlet gauge data it
requires a vanishing source or another charge-conjugation-invariant choice.
If a nonzero source is held fixed without being transformed, the two
configurations belong to different boundary problems and the comparison made
below does not apply. Definition~\ref{def:dressing} is needed because an open
Wilson line is a transition amplitude between endpoint states and not a
gauge-invariant trace by itself, so equality of such amplitudes is a
statement about equally dressed objects.

\begin{remark}
\label{rem:notstationary}
Neither definition invokes stationarity, a radial ansatz, weak backreaction,
self-duality or uniqueness of the gravitational boundary-value problem.
Stationarity enters only in Corollary~\ref{cor:btz}, where the common
asymptotic data are interpreted as BTZ parameters.
\end{remark}

\subsection{Charge conjugation and the gravitational Wilson algebra}
\label{sec:cmap}

The action~\eqref{eq:action} is invariant under the exact charge-conjugation
map
\begin{equation}
{\sf C}:\qquad \phi\mapsto\phi^*,\qquad A_\mu\mapsto-A_\mu,
\qquad n\mapsto-n.
\label{eq:C}
\end{equation}
Under this transformation the gauge-covariant derivative and field strength
satisfy
\begin{equation}
D_\mu\phi\mapsto(D_\mu\phi)^*,\qquad F_{\mu\nu}\mapsto-F_{\mu\nu}.
\label{eq:Cfields}
\end{equation}
The scalar kinetic term and the potential are invariant, and the matter
Chern-Simons density obeys
\begin{equation}
(-A)\wedge\dd(-A)=A\wedge\dd A.
\label{eq:CSC}
\end{equation}
Because the Chern-Simons density is metric independent, the stress tensor is
entirely the scalar stress tensor,
\begin{align}
T_{\mu\nu}={}&(D_\mu\phi)^*D_\nu\phi+(D_\nu\phi)^*D_\mu\phi\nonumber\\
&-g_{\mu\nu}\left[(D_\rho\phi)^*D^\rho\phi-V(|\phi|)\right],
\label{eq:Tformula}
\end{align}
and Eq.~\eqref{eq:Cfields} gives
\begin{equation}
T_{\mu\nu}[\phi,A]=T_{\mu\nu}[\phi^*,-A].
\label{eq:Tparity}
\end{equation}

\begin{theorem}[Charge-conjugation pairing]
\label{thm:pairing}
Let $(g_{\mu\nu},\phi,A_\mu;\mathcal B)_n$ be an admissible branch in the
sense of Definition~\ref{def:admissible}. Then
\begin{equation}
(g_{\mu\nu},\phi,A_\mu;\mathcal B)_n
\xrightarrow{\ {\sf C}\ }
(g_{\mu\nu},\phi^*,-A_\mu;\mathcal B)_{-n}
\label{eq:Cpairing}
\end{equation}
is a second admissible branch, of winding $-n$, with the same metric. In
particular
\begin{equation}
R_{\mu\nu\rho\sigma}[n]=R_{\mu\nu\rho\sigma}[-n].
\label{eq:curvatureblind}
\end{equation}
\end{theorem}

\begin{proof}
It suffices to check that ${\sf C}$ maps each field equation to itself or to
its complex conjugate. The Einstein equation
$G_{\mu\nu}-g_{\mu\nu}/L^2=8\pi G\,T_{\mu\nu}$ is preserved because the
metric is untouched and $T_{\mu\nu}$ is invariant by
Eq.~\eqref{eq:Tparity}. The scalar equation
$D_\mu D^\mu\phi+\partial V/\partial\phi^*=0$ maps to its complex conjugate,
using $D_\mu\phi\mapsto(D_\mu\phi)^*$ from Eq.~\eqref{eq:Cfields} and the
dependence of $V$ on $|\phi|$ alone, which is hypothesis (H1). The gauge
equation is
\begin{equation}
\frac{\kappa}{2}\epsilon^{\mu\nu\rho}F_{\nu\rho}=J^\mu,
\qquad
J^\mu=ie\left[\phi^*D^\mu\phi-\phi\left(D^\mu\phi\right)^*\right].
\label{eq:gaugeEOM}
\end{equation}
Under ${\sf C}$ the left-hand side reverses sign by
Eq.~\eqref{eq:Cfields}, and the right-hand side reverses sign because the two
terms of $J^\mu$ are exchanged. The equation is therefore preserved. Since
$\mathcal B$ is ${\sf C}$-invariant by hypothesis (H3), the image is
admissible. The winding of the scalar phase reverses because
$\phi\propto e^{in\varphi}\mapsto e^{-in\varphi}$. Equality of the metrics is
immediate from Eq.~\eqref{eq:Cpairing}, and Eq.~\eqref{eq:curvatureblind}
follows because every local curvature tensor is a functional of the metric
alone.
\end{proof}

Choose the same local Lorentz gauge for the two charge-conjugate solutions.
Because the metrics coincide, their triads and torsion-free spin connections
can then be chosen identically, and hence
\begin{equation}
\cA^{(\pm)}[n]=\cA^{(\pm)}[-n].
\label{eq:connectionblind}
\end{equation}
For any closed curve $\gamma$, define
\begin{equation}
\cU_\gamma^{(\pm)}(n)=
P\exp\left[-\oint_\gamma\cA^{(\pm)}[n]\right].
\label{eq:gravholonomy}
\end{equation}

\begin{theorem}[Holonomy separation]
\label{thm:holsep}
Under the hypotheses of Theorem~\ref{thm:pairing},
\begin{equation}
\Tr_{\mathcal R}\cU_\gamma^{(\pm)}(n)
=
\Tr_{\mathcal R}\cU_\gamma^{(\pm)}(-n)
\label{eq:gravholonomyblind}
\end{equation}
for every representation $\mathcal R$ and every closed curve $\gamma$.
Moreover, denoting by $\cG_{\rm grav}$ either a closed gravitational Wilson
network or an open or boundary-anchored network supplied with matched
dressing data in the sense of Definition~\ref{def:dressing},
\begin{equation}
\cG_{\rm grav}[n]=\cG_{\rm grav}[-n].
\label{eq:WilsonAlgebraBlind}
\end{equation}
The matter flux meanwhile reverses, $\Phi_n\mapsto-\Phi_n$.
\end{theorem}

\begin{proof}
Equation~\eqref{eq:connectionblind} is an equality of connection one-forms at
every point of the manifold. The path-ordered exponential in
Eq.~\eqref{eq:gravholonomy} is a functional of the connection restricted to
the contour, so the two path-ordered exponentials agree as group elements,
and their characters agree in every representation. For an open or
boundary-anchored network the amplitude is additionally a function of the
endpoint states, representation, boundary frame and junction intertwiners.
These are equal by Definition~\ref{def:dressing}, so the amplitudes agree
path by path. Reversal of the flux is Eq.~\eqref{eq:C}.
\end{proof}

\begin{remark}
\label{rem:noflatness}
Inside the matter-supported core the gravitational connections need not be
flat, so the algebraic path-independence construction used for flat
Chern-Simons backgrounds cannot be applied unchanged to a contour that
crosses the core. The proof above does not use path independence, a
deformation to a geodesic, or a completeness theorem for holonomies in a
sourced geometry. It uses only pointwise equality of the full gravitational
connections along the same contour with the same dressing data. This is why
Theorem~\ref{thm:holsep} strengthens metric blindness into equality of every
closed holonomy character rather than merely of geodesic lengths.
\end{remark}

\subsection{Asymptotic data, Virasoro modes and the BTZ corollary}
\label{sec:btzholonomy}

\begin{corollary}[Ba\~nados functions and Virasoro modes]
\label{cor:banados}
Assume in addition that the matter fields decay fast enough to preserve
standard Brown-Henneaux boundary conditions, so that the asymptotic metric is
parametrized by two Ba\~nados functions $\mathcal L_+(x^+)$ and
$\mathcal L_-(x^-)$~\cite{BrownHenneaux1986,Carlip2023Review}. Then
\begin{equation}
\mathcal L_\pm(x^\pm;n)=\mathcal L_\pm(x^\pm;-n),
\qquad
L_m(n)=L_m(-n),\quad \bar L_m(n)=\bar L_m(-n)
\label{eq:banadosblind}
\end{equation}
for every integer $m$, and in particular
\begin{equation}
\langle T_{++}\rangle_n=\langle T_{++}\rangle_{-n},
\qquad
\langle T_{--}\rangle_n=\langle T_{--}\rangle_{-n}.
\label{eq:boundaryblind}
\end{equation}
\end{corollary}

The equality is therefore not restricted to the two conserved zero modes. It
includes the full classical boundary-graviton data within the chosen
Brown-Henneaux phase space.

\begin{corollary}[BTZ charges, horizons and entropies]
\label{cor:btz}
Assume in addition that the solution is stationary with asymptotic BTZ
connection. The angular-cycle holonomies may then be conjugated to diagonal
representatives with
traces~\cite{BHTZ1993,Carlip1995BTZ,Carlip2005Review}
\begin{align}
\Tr\rho_L&=2\cosh\left[\frac{\pi(r_+-r_-)}{L}\right],\nonumber\\
\Tr\rho_R&=2\cosh\left[\frac{\pi(r_++r_-)}{L}\right],
\label{eq:BTZholtraces}
\end{align}
and the same parameters determine
\begin{equation}
M=\frac{r_+^2+r_-^2}{8GL^2},
\qquad
J=\frac{r_+r_-}{4GL},
\label{eq:BTZcharges}
\end{equation}
together with the Brown-Henneaux zero modes
\begin{equation}
\mathcal L_\pm=\frac{(r_+\pm r_-)^2}{16GL},
\qquad
c=\bar c=\frac{3L}{2G}.
\label{eq:BHdata}
\end{equation}
Then
\begin{equation}
\rho_{L,R}(n)\sim\rho_{L,R}(-n),\qquad
M(n)=M(-n),\quad J(n)=J(-n),\quad r_\pm(n)=r_\pm(-n),
\label{eq:BTZblind}
\end{equation}
where $\sim$ denotes equality of conjugacy classes, and
\begin{equation}
S_{\rm BH}(n)=S_{\rm BH}(-n)=\frac{2\pi r_+}{4G},
\qquad
S_{\rm Cardy}(n)=S_{\rm Cardy}(-n),
\label{eq:entropyblind}
\end{equation}
with the standard Brown-Henneaux identification~\cite{Strominger1998}.
\end{corollary}

The complete boundary states need not coincide. Their distinction can reside
in a matter topological-defect sector that is not generated by the stress
tensor, which is the possibility developed in
Sec.~\ref{sec:interferometer}.

\subsection{Exhaustion of classical Einstein-gravity channels}
\label{sec:exhaustion}

Within classical $2+1$-dimensional Einstein gravity with an invertible triad,
the nontrivial data relevant here fall into three channels: local curvature
fixed by the matter source, global gravitational parallel transport, and
asymptotic boundary modes~\cite{Carlip2005Review,Carlip2023Review}. Equations
\eqref{eq:curvatureblind}, \eqref{eq:connectionblind} and
\eqref{eq:banadosblind} establish equality in all three. Defining
\begin{equation}
\mathfrak D_{\rm grav}
\equiv
\left\{g_{\mu\nu},\cA^{(+)},\cA^{(-)},\mathcal L_+,\mathcal L_-\right\},
\qquad
\mathfrak D_{\rm grav}(n)=\mathfrak D_{\rm grav}(-n),
\label{eq:gravdataexhaustion}
\end{equation}
we conclude that every classical observable that is a functional only of
these common fields, and that uses matched relational or boundary dressing
data, is orientation blind. This conclusion does not assume that closed
holonomies alone reconstruct a sourced geometry. It uses equality of the full
metric and connections before any observable is evaluated.

\subsection{Consequences for geometric and Wilson-line entanglement}
\label{sec:tparity}

\begin{corollary}[Geometric entanglement observables]
\label{cor:geom}
At leading order in $1/G$, RT surfaces~\cite{RyuTakayanagi2006}, HRT
surfaces~\cite{HRT2007}, the EWCS, the
holographic entanglement of purification, the reflected entropy, the mutual
information and the geometric Markov gap are functionals of the same
gravitational fields and boundary regions. Equality of the two metrics
therefore gives
\begin{equation}
S(A;n)=S(A;-n),\qquad E_W(n)=E_W(-n),
\label{eq:geomblind}
\end{equation}
and in AdS$_3$, where the cross section is a
length~\cite{NguyenEtAl2018,TakayanagiUmemoto2018},
\begin{equation}
\EPH(A{:}B;n)\equiv\frac{E_W(A{:}B;n)}{4G}
=\EPH(A{:}B;-n).
\label{eq:purificationblind}
\end{equation}
\end{corollary}

For a connected static wedge, $\EPH$ obeys the standard correlation corridor
\begin{equation}
\frac{1}{2}I(A{:}B;n)\leq \EPH(A{:}B;n)
\leq \min\{S(A;n),S(B;n)\},
\label{eq:purificationbounds}
\end{equation}
and every quantity entering both bounds is separately even under $n\to -n$.
These inequalities therefore constrain the common geometric correlation cost
but contain no information about vortex orientation. As stated in
Sec.~\ref{sec:intro}, $\EPH$ denotes the holographic prescription rather than
the unrestricted field-theoretic minimization.

For a boundary interval in spin-$2$ gravity,
Ref.~\cite{AmmonCastroIqbal2013} gives an equivalent Chern-Simons
description of the RT geodesic as an open Wilson-line amplitude in an
infinite-dimensional highest-weight representation, using the endpoint
conditions $U_i=U_f=\mathbf 1$ and the replica-limit assignment
$\sqrt{2c_2}\to c/6$. With matched dressing data in the sense of
Definition~\ref{def:dressing}, Eq.~\eqref{eq:connectionblind} gives identical
RT Wilson amplitudes for $n$ and $-n$.

The EWCS is generally a bulk-to-bulk geodesic whose endpoints lie on RT
surfaces. A fully gauge-invariant Chern-Simons representation would require
specified bulk endpoint states or a Wilson network with explicit junction
data. Such a construction is not supplied by the single-interval
boundary-anchored prescription and is not assumed in deriving
Eq.~\eqref{eq:geomblind}, whose equality follows from the common metric.
Nevertheless, any proposed gravitational Wilson network for the EWCS would
also be identical for the charge-conjugate pair once its contours, endpoint
states, representations and junction data are fixed identically.

Using $S_R=2\EPH=2E_W/(4G)$ at leading order, both the reflected entropy and
the geometric Markov gap $\Delta_M=S_R-I=2\EPH-I$~\cite{HaydenParrikarSorce2021}
obey
\begin{equation}
\Delta_M(n)=\Delta_M(-n).
\label{eq:markovblind}
\end{equation}
The equality includes stationary frame dragging, because
Eq.~\eqref{eq:Tparity} holds for every component of the stress tensor. It is
not a peculiarity of geodesic observables: any leading semiclassical
observable that is a functional solely of the common gravitational fields, or
of identically dressed gravitational Wilson amplitudes, is equal for the two
orientations. Generic bulk quantum corrections lie outside this classical
pure-gravity corollary. Section~\ref{sec:replicablind} shows, however, that all
charge-conjugation-even fixed-filling replica corrections remain orientation
blind. Sensitivity requires a charged or charge-conjugation-odd insertion,
explicit or anomalous symmetry breaking, or unmatched filling and dressing
data.

\subsection{Euclidean continuation and marked replica fillings}
\label{sec:euclidean}

The holonomy theorem above is Lorentzian, whereas the purification and
reflected-entropy constructions use a Euclidean replica path integral. For a
real Euclidean triad and spin connection, the gravitational variables combine
into the $\mathrm{SL}(2,\mathbb C)$ connections
\begin{equation}
\cA_E=\omega+\frac{i}{L}e,
\qquad
\overline{\cA}_E=\omega-\frac{i}{L}e.
\label{eq:EuclideanConnections}
\end{equation}
The identical Lorentzian metrics of the charge-conjugate pair admit identical
matched continuations, so
\begin{equation}
\cA_E[n]=\cA_E[-n],
\qquad
\overline{\cA}_E[n]=\overline{\cA}_E[-n].
\label{eq:EuclideanBlind}
\end{equation}
For rotating BTZ the regular Euclidean saddle is commonly described by
complexified fields. Equation~\eqref{eq:EuclideanBlind} then refers to the
same analytic-continuation prescription applied to the identical Lorentzian
geometry and boundary data. Charge conjugation still leaves the Euclidean
matter Chern-Simons density invariant while reversing the background flux
seen by a charged topological line.

A Euclidean BTZ exterior is a solid torus. The boundary torus alone does not
specify the bulk saddle: one must choose a marking $\fmark$ that identifies
the contractible cycle. For the line construction, $\fmark$ also denotes the
fixed framing, replica gluing and common boundary edge-sector data required
for Chern-Simons
sewing~\cite{Witten1989,ElitzurEtAl1989,Witten1992Sewing,Carlip2023Review}.
We define
\begin{equation}
\nu_{\fmark}
\equiv
\Lk_{\fmark}(\Gamma_p,\cV_n).
\label{eq:markedlinking}
\end{equation}
All line-amplitude formulas below refer to one fixed marked filling. Large
diffeomorphisms that change the marking relate different bulk fillings and
can change the representative linking data. A nonperturbative modular sum
over such fillings is a different observable and is not evaluated here.

\subsection{Fixed-filling replica blindness beyond the classical saddle}
\label{sec:replicablind}

The equality of the two Euclidean metrics is sufficient for the leading
geometric entropy, but the symmetry gives a stronger path-integral identity.
Let $\mathfrak F_{q,\fmark}(n)$ be the gauge-fixed integration domain on an
integer-$q$ replica filling with marking $\fmark$ and boundary conditions that
prepare winding $n$. For sources $\mathcal J_+$ coupled only to
charge-conjugation-even operators, define
\begin{equation}
Z_{q,\fmark}[n;\mathcal J_+]
=\int_{\mathfrak F_{q,\fmark}(n)}\!\mathcal D\mu_{q,\fmark}\,
\exp\!\left[-I_E+\sum_a\int \mathcal J_{+,a}\mathcal O_{+,a}\right],
\label{eq:replicapartition}
\end{equation}
where $\mathcal D\mu_{q,\fmark}$ includes the gravitational and matter gauge
fixing, ghosts and the chosen edge-mode or operator-algebra prescription.

\begin{definition}[Quantum-admissible replica problem]
\label{def:quantumadmissible}
The fixed-filling replica problem in Eq.~\eqref{eq:replicapartition} is called
\emph{quantum admissible} if: (Q1) charge conjugation maps the gauge-fixed
integration cycle $\mathfrak F_{q,\fmark}(n)$ bijectively onto
$\mathfrak F_{q,\fmark}(-n)$, including any complexified cycle used for
rotating saddles, without changing the marking, gluing or even sources;
(Q2) the gauge fixing, functional measure and regulator are
charge-conjugation invariant, so there is no charge-conjugation anomaly;
(Q3) the same charge-conjugation-covariant algebra and gauge-edge prescription
is used in the two sectors; and (Q4) whenever a noninteger replica number is
required, the continuation in $q$ is performed by the same prescription in the
two sectors.
\end{definition}

\begin{theorem}[Fixed-filling replica blindness]
\label{thm:replicablind}
For every quantum-admissible fixed filling, every positive integer $q$ and
every collection of charge-conjugation-even sources,
\begin{equation}
Z_{q,\fmark}[n;\mathcal J_+]
=Z_{q,\fmark}[-n;\mathcal J_+].
\label{eq:replicablind}
\end{equation}
The equality is a change-of-variables identity and is not restricted to the
classical saddle or to one-loop order.
\end{theorem}

\begin{proof}
Charge conjugation acts on the field multiplet as the linear involution
\begin{equation}
(\operatorname{Re}\phi,\operatorname{Im}\phi,A)
\longmapsto
(\operatorname{Re}\phi,-\operatorname{Im}\phi,-A),
\label{eq:Cinvolution}
\end{equation}
and leaves the metric unchanged. It squares to the identity and is block
diagonal with entries $\pm1$ in any basis adapted to this splitting, so it is
invertible with unit-modulus Jacobian, and by (Q1) it carries the winding-$n$
integration cycle onto the winding-$-n$ cycle. Only invertibility, the unit Jacobian and
preservation of the cycle are used, so the step is valid on a real Euclidean
section and equally on the complexified cycles required for rotating saddles.
The Euclidean action is invariant because the scalar kinetic and potential
terms are invariant and $(-A)\wedge\dd(-A)=A\wedge\dd A$, the latter holding
irrespective of whether the Chern-Simons contribution is real or imaginary in
the Euclidean continuation. Condition (Q2) makes the regulated Jacobian, ghost
measure and counterterm scheme invariant, while every source insertion is
unchanged by construction. Changing variables in
Eq.~\eqref{eq:replicapartition} gives Eq.~\eqref{eq:replicablind}.
\end{proof}

For an $A$-branched replica, define the normalized R\'enyi entropy by
\begin{equation}
S_{q,\fmark}(A;n)=\frac{1}{1-q}
\log\!\left[\frac{Z_{q,\fmark}[n]}{Z_{1,\fmark}[n]^q}\right].
\label{eq:renyifixedfilling}
\end{equation}
Equation~\eqref{eq:replicablind} immediately gives
\begin{equation}
S_{q,\fmark}(A;n)=S_{q,\fmark}(A;-n),
\label{eq:renyiblind}
\end{equation}
for every integer $q>1$, and, by (Q4), for the von Neumann entropy obtained by
continuation to $q=1$. The same argument applies to uncharged
reflected-replica and multipartite replica constructions, provided their
gluing data satisfy Definition~\ref{def:quantumadmissible}.

\begin{corollary}[FLM and quantum-extremal-surface blindness]
\label{cor:flmblind}
Let $X$ be any candidate codimension-two curve in the common Euclidean
geometry and let $\Sigma_X$ be its associated bulk region. Through one loop,
the renormalized generalized entropy may be organized as
\begin{equation}
S_{\rm gen}[X;n]=\frac{\operatorname{Length}(X)}{4G_{\rm ren}}
+S_{\rm bulk}^{\rm alg}(\Sigma_X;n)+S_{\rm loc}[X;n],
\label{eq:generalizedentropy}
\end{equation}
where $S_{\rm loc}$ contains the local Wald-like and counterterm contributions
of the FLM expansion~\cite{FaulknerLewkowyczMaldacena2013}. Under the
hypotheses of Theorem~\ref{thm:replicablind},
\begin{equation}
S_{\rm gen}[X;n]=S_{\rm gen}[X;-n]
\quad\hbox{for every candidate }X.
\label{eq:genblind}
\end{equation}
Consequently the two sectors have the same set of quantum extremal surfaces
and the same minimum generalized entropy~\cite{EngelhardtWall2015}. Their
full gauge-fixed one-loop effective actions also agree,
\begin{equation}
\Gamma^{(1)}_{q,\fmark}[n]=\Gamma^{(1)}_{q,\fmark}[-n].
\label{eq:oneloopblind}
\end{equation}
\end{corollary}

The determinant pairing and the treatment of the bulk algebra are detailed in
Appendix~\ref{sec:replicaapp}. If one later sums over modularly related
fillings, the equality survives term by term when the set of fillings and
their weights are themselves charge-conjugation invariant.

\subsection{Orientation resolution in the charged sector}
\label{sec:chargedsector}

Theorem~\ref{thm:replicablind} is sharp, and the same change of variables that
closes the neutral sector determines the charged one exactly. Let $\cQ_A$
denote the conserved matter charge of the entangling region, fixed on the
replica by the Chern-Simons Gauss law of Eq.~\eqref{eq:gauss}, and let $\mu$
be the conjugate Aharonov-Bohm angle inserted around the branch locus, so that
\begin{equation}
Z_{q,\fmark}[n;\mu]
=\int_{\mathfrak F_{q,\fmark}(n)}\!\mathcal D\mu_{q,\fmark}\,
e^{-I_E+i\mu\cQ_A}.
\label{eq:chargedreplica}
\end{equation}
Charged replica partition functions of this form are the standard generating
functionals of symmetry-resolved
entanglement~\cite{GoldsteinSela2018,BelinEtAl2013,ZhaoNortheMeyer2021}.

\begin{corollary}[Charged replica pairing and orientation resolution]
\label{cor:chargedblind}
Under the hypotheses of Theorem~\ref{thm:replicablind},
\begin{equation}
Z_{q,\fmark}[n;\mu]=Z_{q,\fmark}[-n;-\mu].
\label{eq:chargedpairing}
\end{equation}
Writing $\log Z_{q,\fmark}[n;\mu]=E_q[n;\mu]+O_q[n;\mu]$ with $E_q$ even and
$O_q$ odd in $\mu$, this is equivalent to
\begin{equation}
E_q[n;\mu]=E_q[-n;\mu],
\qquad
O_q[n;\mu]=-O_q[-n;\mu].
\label{eq:evenoddsplit}
\end{equation}
The symmetry-resolved partition functions
$\cZ_q(\cQ;n)=\int\!\frac{\dd\mu}{2\pi}\,e^{-i\mu\cQ}Z_{q,\fmark}[n;\mu]$ and
the entropies built from them satisfy
\begin{equation}
\cZ_q(\cQ;n)=\cZ_q(-\cQ;-n),
\qquad
S_q(\cQ;n)=S_q(-\cQ;-n),
\label{eq:symmetryresolved}
\end{equation}
and every odd charged cumulant, in particular
$\langle\cQ_A\rangle_{q,n}
=-i\,\partial_\mu\log Z_{q,\fmark}[n;\mu]\big|_{\mu=0}$, is exactly odd in
$n$.
\end{corollary}

\begin{proof}
Charge conjugation reverses the matter charge, $\cQ_A\to-\cQ_A$, and carries
the winding-$n$ cycle onto the winding-$-n$ cycle as in
Theorem~\ref{thm:replicablind}. The insertion $e^{i\mu\cQ_A}$ is therefore
carried to $e^{-i\mu\cQ_A}$, and the change of variables gives
Eq.~\eqref{eq:chargedpairing}. Separating $\log Z_{q,\fmark}$ into its parts
even and odd in $\mu$ and matching the two sides gives
Eq.~\eqref{eq:evenoddsplit}. Fourier transforming
Eq.~\eqref{eq:chargedpairing} in $\mu$ gives the first equality of
Eq.~\eqref{eq:symmetryresolved}, and the second follows because the
symmetry-resolved entropies are functionals of $\cZ_q(\cQ;n)$ alone. The
cumulant statement is the leading term of Eq.~\eqref{eq:evenoddsplit} in
$\mu$.
\end{proof}

Equations~\eqref{eq:renyiblind} and~\eqref{eq:symmetryresolved} together give
a complete accounting of where the orientation is stored. Summing
Eq.~\eqref{eq:symmetryresolved} over $\cQ$ reconstructs the total entropy,
which is blind by Eq.~\eqref{eq:renyiblind}, so the entire dependence on
$\operatorname{sgn}(n)$ resides in the charge asymmetry of the entanglement
spectrum and in no other feature of it. The leading value of that odd datum is
fixed by the topological content of the sector: for an entangling region whose
wedge contains the vortex core, the Gauss law gives
\begin{equation}
\langle\cQ_A\rangle_{1,n}=\kappa\Phi_n=\frac{2\pi\kappa n}{e},
\label{eq:oddcumulantvalue}
\end{equation}
which reverses with the winding and vanishes only in the trivial sector. The
linked line amplitude constructed in Sec.~\ref{sec:interferometer} realizes
the same odd datum at topological order. Equation~\eqref{eq:chargedpairing}
shows that the geometric blindness and the matter-sector readout are the even
and odd halves of a single identity rather than independent constructions.

\section{The purification gate: exact BTZ cross sections}
\label{sec:purificationgate}

The holonomy theorem is exact but abstract. The analytic BTZ cross sections
of Ref.~\cite{NguyenEtAl2018} turn it into a quantitative mixed-state
statement, and expose a distinction that is essential for the readout: the
value of the geometric purification cost, and the availability of a
cross-section-supported contour, are separate from the phase carried by the
matter line. We refer to the second of these throughout as the
\emph{purification gate}. Section~\ref{sec:erasure} gives it a restricted
code-theoretic interpretation as a geometric access threshold for the
specified cross-section-supported line protocol; no claim of complete quantum
information erasure is made.

\subsection{One-sided BTZ: an orientation-blind purification plateau}
\label{sec:onesidedpurification}

Consider a nonrotating one-sided BTZ geometry and complementary boundary
regions with half-widths $\alpha$ and $\pi-\alpha$. Let $r_c$ be a radial
cutoff and define
\begin{equation}
\alpha_c=\frac{L}{r_+}\arcsinh(1).
\label{eq:alphacpur}
\end{equation}
For the thermodynamically relevant large-black-hole branch the holographic
purification cost is~\cite{NguyenEtAl2018}
\begin{equation}
\EPH(\alpha;r_+)=\frac{L}{2G}
\begin{cases}
\displaystyle
\log\!\left[\frac{2r_c}{r_+}
\sinh\!\left(\frac{r_+\alpha}{L}\right)\right],
&0<\alpha<\alpha_c,\\[6pt]
\displaystyle
\log\!\left(\frac{2r_c}{r_+}\right),
&\alpha_c\leq\alpha\leq\pi-\alpha_c,\\[6pt]
\displaystyle
\log\!\left[\frac{2r_c}{r_+}
\sinh\!\left(\frac{r_+(\pi-\alpha)}{L}\right)\right],
&\pi-\alpha_c<\alpha<\pi.
\end{cases}
\label{eq:btzpurplateau}
\end{equation}
The three branches match continuously at $\alpha=\alpha_c$ and
$\alpha=\pi-\alpha_c$, since $\sinh(r_+\alpha_c/L)=1$ by
Eq.~\eqref{eq:alphacpur}. The middle branch is the entanglement-purification
plateau. It is determined by two radial geodesics and is independent of
$\alpha$.

The plateau exists only when the two matching points do not cross, that is
when $\alpha_c\leq\pi/2$. By Eq.~\eqref{eq:alphacpur} this is the single
inequality
\begin{equation}
\frac{r_+}{L}\ \geq\ \frac{2}{\pi}\arcsinh(1)
=\frac{2}{\pi}\log\left(1+\sqrt2\right)\simeq0.5611,
\label{eq:plateaucondition}
\end{equation}
which is the precise sense in which Eq.~\eqref{eq:btzpurplateau} refers to a
large black hole. The values $r_+/L=1,2,5$ used in
Fig.~\ref{fig:purificationgate} all satisfy it.

Since charge conjugation leaves $r_+$, $L$ and the cutoff identification
unchanged, Eq.~\eqref{eq:btzpurplateau} is exactly the same for $n$ and $-n$.
The plateau is therefore a strong geometric correlation feature that is
strictly blind to the sign of the anyonic hair. In a fixed marked linked
filling, the matter factor $M_{p,n}^{\nu_{\fmark}}$ is likewise independent
of $\alpha$ in the topological limit, but is inverted by $n\to -n$. The
plateau thus cleanly separates a real geometric cost from a quantized
topological phase.

\subsection{Two-sided BTZ: connectivity as the purification gate}
\label{sec:twosidedpurification}

The time-dependent two-sided BTZ example gives a sharper test. Let
$z_H=L^2/r_+$ and take $A$ and $B$ to be opposite half-circles on the two
boundaries at equal boundary time $T_0$. While the entanglement wedge is
connected, the holographic purification cost is~\cite{NguyenEtAl2018}
\begin{equation}
\EPH(T_0)=\frac{L}{4G}\arccosh\!\left\{
1+\left[\cosh\!\left(\frac{\pi L}{z_H}\right)-1\right]
\operatorname{sech}^{2}\!\left(\frac{T_0}{z_H}\right)
\right\},
\label{eq:twosidedEP}
\end{equation}
and
\begin{equation}
\frac{I(A{:}B;T_0)}{2}=\frac{L}{2G}
\log\!\left[
\sinh\!\left(\frac{\pi L}{2z_H}\right)
\operatorname{sech}\!\left(\frac{T_0}{z_H}\right)
\right].
\label{eq:twosidedMI}
\end{equation}
The connected saddle loses dominance at
\begin{equation}
T_*=z_H\arccosh\!\left[
\sinh\!\left(\frac{\pi L}{2z_H}\right)\right].
\label{eq:Tstar}
\end{equation}
Since $z_H=L^2/r_+$, reality of Eq.~\eqref{eq:Tstar} requires
$\sinh(\pi r_+/2L)\geq1$, which is again exactly the inequality
\eqref{eq:plateaucondition}. The same constant $\arcsinh(1)$ therefore
controls the existence of the one-sided plateau and the existence of the
two-sided transition.

As $T_0\to T_*^-$ the leading geometric mutual information approaches zero continuously,
whereas $\EPH$ remains finite and then jumps to zero when the wedge
disconnects. The height of that jump is fixed and universal. Using
$\cosh(\pi L/z_H)-1=2\sinh^2(\pi L/2z_H)$ together with
$\cosh(T_*/z_H)=\sinh(\pi L/2z_H)$ from Eq.~\eqref{eq:Tstar}, the brace in
Eq.~\eqref{eq:twosidedEP} equals $3$ identically, independently of $L$, $z_H$
and $G$, so that
\begin{equation}
\EPH(T_*^-)=\frac{L}{4G}\arccosh 3
=\frac{L}{2G}\arcsinh 1
=\frac{L}{2G}\log\left(1+\sqrt2\right)
=\frac{c}{3}\log\left(1+\sqrt2\right),
\label{eq:universaljump}
\end{equation}
where $c=3L/2G$ is the Brown-Henneaux central charge and we used
$\arccosh 3=2\arcsinh 1$. Unlike the one-sided plateau height in
Eq.~\eqref{eq:btzpurplateau}, which depends on the radial cutoff $r_c$,
Eq.~\eqref{eq:universaljump} is cutoff independent, because $A$ and $B$ live
on different boundaries and share no endpoints. The gate therefore opens and
closes at a scheme-independent height set by the central charge alone. Because
the mutual information vanishes continuously at the same point, the geometric
Markov gap $\Delta_M=2\EPH-I$ approaches an equally universal limit,
\begin{equation}
\Delta_M(T_*^-)=\frac{2c}{3}\log\left(1+\sqrt2\right)
=\frac{c}{3}\log\left(3+2\sqrt2\right),
\label{eq:universalmarkov}
\end{equation}
consistent with the general Markov gap bound of
Ref.~\cite{HaydenParrikarSorce2021}. A derivation of both universal numbers is
given in Appendix~\ref{sec:btzpurapp}.

This nonanalyticity supplies the operational gate for our protocol. For
$0\leq T_0<T_*$, a chosen cross section can be doubled into the closed probe
contour, and if the marked replica filling has linking number $\nu_{\fmark}$
its vortex-antivortex contrast is
\begin{equation}
\cX_p^{(\fmark)}(n;T_0)=\exp[-2i\nu_{\fmark}\kappa\Phi_p\Phi_n]
\left[1+O(e^{-d/\xi})\right].
\label{eq:timeindependentcontrast}
\end{equation}
The phase is independent of $T_0$ in the topological limit even though the
geometric purification cost changes continuously. At $T_0\geq T_*$ the
specific cross-section-supported channel ceases to exist. This does not set
the anyonic phase to zero; it closes the protocol that accesses it. Wedge
connectivity therefore controls availability, while charge conjugation
controls the sign of the matter holonomy.

The theorem of Sec.~\ref{sec:replicablind} fixes the first quantum correction
to this picture as well. When the classical RT surface for $A\cup B$ is the
disconnected union of the surfaces for $A$ and $B$, the leading nonvanishing
contribution to the mutual information is the bulk mutual information at order
$G^0$~\cite{FaulknerLewkowyczMaldacena2013},
\begin{equation}
I(A{:}B)=I_{\rm bulk}^{\rm alg}(A_b{:}B_b)+O(G),
\label{eq:flmresidualMI}
\end{equation}
and Theorem~\ref{thm:replicablind} gives the sector statement
\begin{equation}
I_{\rm bulk}^{\rm alg}(A_b{:}B_b;n)
=I_{\rm bulk}^{\rm alg}(A_b{:}B_b;-n).
\label{eq:flmresidualblind}
\end{equation}
Whatever neutral correlation survives the disconnection is therefore common to
the two orientations, so the gate acts on the geometric and the one-loop
channel alike. By Corollary~\ref{cor:chargedblind} the orientation is carried
instead by the charge-odd data of the same replica, of which the linked line
amplitude is the topological-limit representative.

\begin{figure}[t]
\centering
\includegraphics[width=0.49\textwidth]{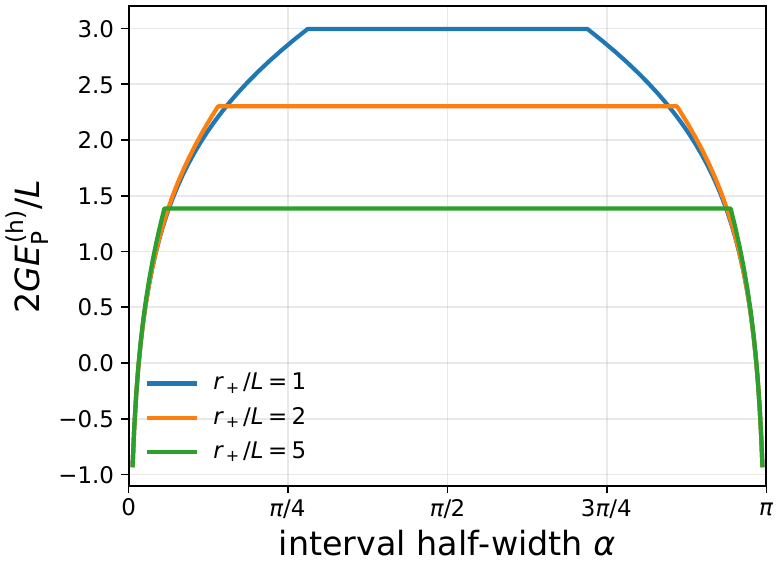}\hfill
\includegraphics[width=0.49\textwidth]{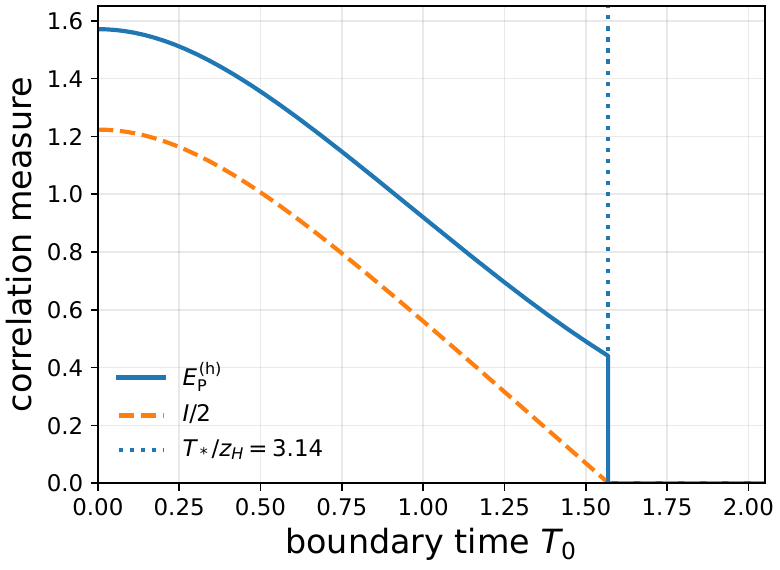}
\caption{Analytic BTZ purification benchmarks reconstructed from
Ref.~\cite{NguyenEtAl2018}. Left: the one-sided holographic purification cost
for complementary intervals, showing the horizon-supported plateau for
$r_+/L=1,2,5$ at $r_c/L=10$, all of which satisfy
Eq.~\eqref{eq:plateaucondition}. Right: the two-sided time-dependent $\EPH$
and $I/2$ for $L=1$, $z_H=1/2$ and $G=1$. The vertical line marks $T_*$, at
which the leading geometric $I$ vanishes continuously while $\EPH$ jumps from the universal value
of Eq.~\eqref{eq:universaljump} to zero. Every geometric curve is identical
for charge-conjugate vortex sectors. In the proposed linked protocol the
matter-sector contrast remains the quantized phase in
Eq.~\eqref{eq:timeindependentcontrast} throughout the connected phase, and
becomes unavailable when the wedge disconnects.}
\label{fig:purificationgate}
\end{figure}

\subsection{The gravitating exterior: a conditional linearized charge map}
\label{sec:exterior}

Theorems~\ref{thm:holsep} and~\ref{thm:replicablind} show exactly that the
neutral gravitational sector cannot resolve $\operatorname{sgn}(n)$. The
flat-space BPS profiles allow a separate, more quantitative question: if a
regular stationary asymptotically BTZ branch exists and is continuously
connected to a localized vortex core, what $O(G)$ asymptotic charge shifts
would that core induce? The calculation in this subsection answers only that
conditional perturbative question.

Because the vortex carries angular momentum, the appropriate circularly
symmetric ansatz is stationary rather than static,
\begin{equation}
\dd s^2=-e^{2\delta(r)}F(r)\,\dd t^2+\frac{\dd r^2}{F(r)}
+r^2\bigl(\dd\varphi+\omega(r)\dd t\bigr)^2 .
\label{eq:stationaryansatz}
\end{equation}
At linear order in $G$, $\omega=O(G)$ while its contribution to the diagonal
metric function is quadratic, so the mass constraint may be evaluated in the
nonrotating sector at this order. The momentum constraint separately fixes
the asymptotic $1/r^2$ coefficient of $\omega$ by the total angular momentum.
It is useful to keep this charge bookkeeping distinct from the physical ADM
normalization used in Eq.~\eqref{eq:BTZcharges}. We therefore define
\begin{equation}
\widehat M\equiv 8G M_{\rm ADM},\qquad
\widehat j\equiv \frac{8GJ_{\rm ADM}}{L},
\qquad
\frac{r_+\pm r_-}{L}=\sqrt{\widehat M\pm\widehat j}.
\label{eq:dimensionlessBTZcharges}
\end{equation}
Both $\widehat M$ and $\widehat j$ are dimensionless. This convention is used
only in Secs.~\ref{sec:exterior}--\ref{sec:band} and in
Appendix~\ref{sec:rotatingplateau}.

For the diagonal constraint, writing $T^t{}_t=-\rho$ gives at $O(G)$
\begin{equation}
F(r)=\frac{r^2}{L^2}-\widehat M_0
-16\pi G\int_0^r u\,\rho(u)\,\dd u+O(G^2),
\label{eq:Fintegrated}
\end{equation}
so the unresolved source outside radius $r$ is measured by the BPS energy-tail
fraction
\begin{equation}
\Delta_n(r)\equiv\frac{1}{|n|}\int_{mr}^\infty y\,
\frac{T_{00}(y)}{m^2v^2}\,\dd y .
\label{eq:tailfraction}
\end{equation}
By Eq.~\eqref{eq:exponentialfalloff}, $\Delta_n$ is exponentially small at
fixed $n$ and large $mr$. Numerically $\Delta_1=1.2\times10^{-11}$,
$\Delta_2=9.8\times10^{-10}$ and $\Delta_3=6.5\times10^{-8}$ at $mr=15$,
falling to $5.3\times10^{-16}$, $4.4\times10^{-14}$ and $2.9\times10^{-12}$
at $mr=20$. Panels (a) and (b) of Fig.~\ref{fig:exteriorband} display the
localization and the corresponding tail fraction.

\begin{proposition}[Conditional linearized charge matching]
\label{prop:exteriorimprint}
Assume that a regular stationary asymptotically BTZ Einstein-CSH branch exists,
that its $O(G)$ integrated charge shifts are controlled by the flat-space BPS
energy and angular momentum, and that $\Delta_n(r_+)\ll1$. Then
\begin{align}
\widehat M_n&=\widehat M_0+\gamma|n|+O(G^2),\nonumber\\
\widehat j_n&=\widehat j_0
-\operatorname{sgn}(\kappa)\frac{\gamma}{mL}\,n^2+O(G^2),
\qquad \gamma\equiv16\pi Gv^2 .
\label{eq:exteriorcharges}
\end{align}
For a nonrotating seed, $\widehat j_0=0$ and therefore
$|\widehat j_n|=\gamma n^2/(mL)$. The two induced shifts are even in $n$.
\end{proposition}

\begin{proof}
At linear order the dimensionless ADM mass shift is
$\delta\widehat M=8G E_n$. Using Eq.~\eqref{eq:quantized} gives
$\delta\widehat M=16\pi Gv^2|n|=\gamma|n|$. Likewise the dimensionless
angular-momentum shift is $\delta\widehat j=8Gs_n/L$. Equations
\eqref{eq:spin} and~\eqref{eq:quantized}, together with
$m=2e^2v^2/|\kappa|$, give
$\delta\widehat j=-\operatorname{sgn}(\kappa)\gamma n^2/(mL)$. The same
$n$ and $n^2$ dependences are independently reproduced by the energy and spin
sum rules in Eq.~\eqref{eq:sumrules}. The tail fraction in
Eq.~\eqref{eq:tailfraction} measures the part of the flat-space source not
captured inside $r_+$; higher gravitational corrections begin at the next
order in the perturbative expansion.
\end{proof}

Proposition~\ref{prop:exteriorimprint} should not be read as a construction of
the missing black-hole branch. It is a charge-matching statement conditional
on that branch. Within those assumptions, the asymptotic $O(G)$ imprint is
fixed by $|n|$ and $n^2$, while residual finite-radius profile dependence is
controlled by the matter tail. The inverse mass scale $1/m$ fixes the relative
coefficient of the two BPS charge shifts; it should not be identified with the
full radius of a multiwinding vortex, whose half-mass radius grows
approximately as $2|n|/m$ in the numerical profiles.

For the exact orientation-blindness result, none of these perturbative
assumptions is needed: charge conjugation maps $n$ to $-n$ while leaving the
metric unchanged, so any admissible stationary branch has the same
$\widehat M$ and $\widehat j$ for the two orientations. The linearized charge
map simply makes explicit which even winding combinations occur in the BPS
core estimate.

\subsection{Conditional continuous winding envelope}
\label{sec:band}

For a nonrotating seed, Eq.~\eqref{eq:exteriorcharges} predicts a mass shift
linear in $|n|$ and an induced dimensionless spin quadratic in $n$. To explore
how those two shifts compete, we use the rotating-BTZ matching diagnostic
recorded in Appendix~\ref{sec:rotatingplateau},
\begin{equation}
\Xi(\widehat M,\widehat j)\equiv
\sinh\!\left(\frac{\pi}{2}\sqrt{\widehat M-|\widehat j|}\right)
\sinh\!\left(\frac{\pi}{2}\sqrt{\widehat M+|\widehat j|}\right)\geq1 .
\label{eq:rotatinggate}
\end{equation}
Its nonrotating limit is exact and gives
\begin{equation}
M_c\equiv\frac{4}{\pi^2}\log^2\!\left(1+\sqrt2\right)
=0.314833044295 .
\label{eq:Mc}
\end{equation}
Reality of the first square root is the BTZ nonextremality condition
$|\widehat j|\leq\widehat M$. For $\widehat j\neq0$, Eq.~\eqref{eq:rotatinggate}
is used here as a conditional continuation of the nonrotating cross-section
matching, not as an independently proved covariant-EWCS theorem.

With
\begin{equation}
u\equiv\gamma|n|,
\qquad
\lambda\equiv\gamma mL=16\pi Gv^2mL ,
\label{eq:banddimensionless}
\end{equation}
Eq.~\eqref{eq:rotatinggate} for a nonrotating seed becomes
\begin{equation}
\sinh\!\left(\frac{\pi}{2}\sqrt{\widehat M_0+u-\frac{u^2}{\lambda}}\right)
\sinh\!\left(\frac{\pi}{2}\sqrt{\widehat M_0+u+\frac{u^2}{\lambda}}\right)
\geq1 .
\label{eq:bandcondition}
\end{equation}
This inequality defines a \emph{continuous candidate envelope} in $u$. The
actual winding remains integer, $n\in\mathbb Z$, so for fixed $\gamma$ the
physical candidates are only those discrete values $u=\gamma|n|$ that lie in
the envelope and also satisfy the localization and branch-existence
conditions of Proposition~\ref{prop:exteriorimprint}.

For $\widehat M_0<M_c$ and large $\lambda$, the continuous lower edge obeys
\begin{equation}
u_-\longrightarrow M_c-\widehat M_0,
\qquad
|n|_{\rm lower}\longrightarrow\frac{M_c-\widehat M_0}{\gamma},
\label{eq:loweredge}
\end{equation}
while the nonextremality boundary gives the exact algebraic upper edge
\begin{equation}
|n|_{\rm ext}=\frac{mL}{2}
\left[1+\sqrt{1+\frac{4\widehat M_0}{\gamma mL}}\right],
\qquad
|n|_{\rm ext}\big|_{\widehat M_0=0}=mL .
\label{eq:nmax}
\end{equation}
For $\widehat M_0=0$, the numerical continuous edges satisfy
$u_-=0.3148331$ at $\lambda=1000$ and $u_+/\lambda=1.0000000$ already at
$\lambda=30$, reproducing these limits.

A second control condition is essential. The extremality edge in
Eq.~\eqref{eq:nmax} does not by itself guarantee that the BPS core is localized
behind the horizon. The numerical profiles give
$R_{1/2}\simeq2|n|/m$. At the massless-seed algebraic edge
$|n|=mL$, this becomes $R_{1/2}\simeq2L$, whereas extremal BTZ has
$r_+/L=\sqrt{\lambda/2}$. Thus, for example, the upper edge in the range
$\lambda\leq6$ displayed in Fig.~\ref{fig:exteriorband}(c) does not yet contain
even half of the BPS energy. The controlled subset of the candidate envelope
is therefore its intersection with a small-tail requirement such as
$\Delta_n(r_+)\ll1$; the figure intentionally does not promote the algebraic
upper edge to a constructed black-hole solution.

Maximizing the left side of Eq.~\eqref{eq:bandcondition} over continuous $u$
gives the envelope threshold
\begin{equation}
\lambda\geq\lambda_\star(\widehat M_0),
\qquad
\lambda_\star(0)=0.581179706 .
\label{eq:lambdastar}
\end{equation}
At $\widehat M_0=0$ the continuous threshold occurs at
$|\widehat j|/\widehat M=0.7534794$, corresponding to
$r_+/r_-=2.19976$. Representative continuous thresholds are
$\lambda_\star=0.334083$, $0.139111$ and $0.015128$ for
$\widehat M_0=0.10$, $0.20$ and $0.30$, and they tend to zero as
$\widehat M_0\to M_c$. Because $n$ is discrete, these values are necessary
thresholds for a nonempty continuous envelope, not sufficient thresholds for
an allowed integer winding at fixed $\gamma$.

\begin{figure}[t]
\centering
\includegraphics[width=0.99\textwidth]{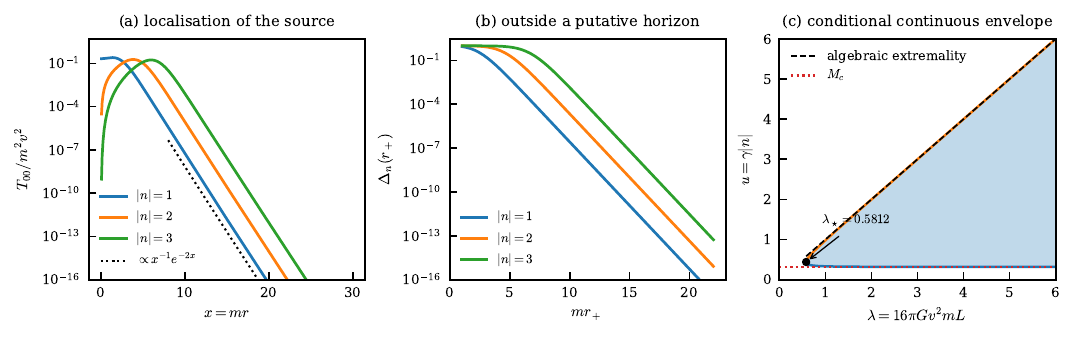}
\caption{Linearized BPS-core diagnostics for a conditional vortex-dressed BTZ
exterior. (a) Self-dual CSH energy density on a logarithmic scale, confirming
the $x^{-1}e^{-2x}$ falloff of Eq.~\eqref{eq:exponentialfalloff}. (b) Fraction
$\Delta_n(r_+)$ of the flat-space BPS energy outside a putative horizon,
Eq.~\eqref{eq:tailfraction}. (c) Continuous candidate envelope from
Eq.~\eqref{eq:bandcondition} at $\widehat M_0=0$ in the plane
$\lambda=16\pi Gv^2mL$ and $u=\gamma|n|$. The dashed edge is the algebraic
nonextremality boundary. Panel (c) does not impose $n\in\mathbb Z$ and does
not impose the small-tail condition $\Delta_n(r_+)\ll1$; those restrictions
must be applied before interpreting any point as a physical black-hole
candidate. The continuous envelope first appears at
$\lambda_\star=0.5812$.}
\label{fig:exteriorband}
\end{figure}

The status of the gravitating vortex branch on which this application rests
is stated in full in Appendix~\ref{sec:scope}. In particular, the flat-space
profiles do not satisfy the full stationary horizon boundary condition, and
no fully coupled positive-mass nonextremal Einstein-CSH branch is constructed
here. The charge-conjugation theorem of Sec.~\ref{sec:cthm} is
unconditional; Secs.~\ref{sec:exterior} and~\ref{sec:band} are explicitly
conditional perturbative diagnostics.

\section{Reading the orientation: the channel-selective interferometer}
\label{sec:interferometer}

\subsection{Reflected-replica geometry and matter-sector probe}
\label{sec:replica}

The holonomy theorem identifies what must be added to a gravitational
purification observable in order to resolve the orientation. The same cross
section $X$ computes $\EPH$ in the holographic purification proposal and one
half of the leading reflected entropy in the canonical purification. The
gravitational Chern-Simons connections determine the replica saddle and its
geometric contribution, but they are identical for $n$ and $-n$ by
Theorem~\ref{thm:holsep}. The distinguishing operator must therefore live in
the matter topological sector.

Let $\cV_n$ denote the closed background vortex-sector line in the fixed
marked Euclidean reflected-replica filling $\fmark$ of
Sec.~\ref{sec:euclidean}. The smooth Euclidean BTZ exterior is a solid torus,
so $\cV_n$ is not identified with a literal trajectory through the Lorentzian
point $r=0$; it is the core cycle of that solid torus. A probe line carrying
vortex sector $p$ is inserted along the EWCS in one replica of the canonical
purification, so that after the reflection gluing of
Ref.~\cite{DuttaFaulkner2021} it forms a single closed contour $\Gamma_p$ in
the reflected replica (Fig.~\ref{fig:linking}). We define the normalized
matter-sector amplitude
\begin{equation}
\cW_p^{(\fmark)}(n)=
\frac{Z_R^{(\fmark)}[\Gamma_p,\cV_n]\,Z_R^{(\fmark)}[0,0]}
{Z_R^{(\fmark)}[\Gamma_p,0]\,Z_R^{(\fmark)}[0,\cV_n]}.
\label{eq:normalized}
\end{equation}

Replica line insertions are the natural charged probes in Chern-Simons
entanglement calculations~\cite{BerthiereEtAl2021,ZhaoNortheMeyer2021}. At
fixed marked filling $\fmark$, including framing, boundary topological
sector, replica gluing and edge-sector data, this ratio removes the separate
probe and background self-energies, local core renormalizations and factors
associated with either line alone. The fixed boundary-sector qualification is
essential because Chern-Simons path integrals glue through boundary
edge-state Hilbert spaces~\cite{ElitzurEtAl1989,Witten1992Sewing}. In the
topological limit of a linked replica saddle, the remaining infrared datum is
the mutual monodromy of the two Abelian vortex sectors. The construction is
therefore channel selective: the gravitational Wilson algebra fixes the
common geometry, while the matter line supplies the orientation-sensitive
phase.

\begin{remark}[Well-definedness of the linking number]
\label{rem:linking}
A linking number of two closed curves in a solid torus is not defined by
homology alone, since the first homology group is nontrivial. Equation
\eqref{eq:markedlinking} is well defined because the marking $\fmark$ fixes
which cycle is contractible, and because $\Gamma_p$, being the reflection
double of a cross section anchored on RT surfaces, is null-homologous in the
complement of $\cV_n$ within the chosen filling. The signed intersection of
$\cV_n$ with any Seifert surface bounded by $\Gamma_p$ in that complement is
then an integer independent of the choice of surface, and equals
$\nu_{\fmark}$. Framing dependence is separately fixed by $\fmark$, as
required for Chern-Simons line
observables~\cite{Witten1989,ElitzurEtAl1989,Witten1992Sewing}.
\end{remark}

\begin{figure}[t]
\centering
\begin{tikzpicture}[scale=1.5,>=Latex]
\draw[thick] (0,0) circle (1.55);
\filldraw[red] (0,0) circle (0.065);
\node[red, font=\footnotesize] at (0.16,-0.22) {$\cV_n$};
\draw[very thick, blue] ({1.55*cos(15)},{1.55*sin(15)}) arc (15:105:1.55)
    node[pos=0.5, above, font=\footnotesize] {$A$};
\draw[very thick, blue] ({1.55*cos(-15)},{1.55*sin(-15)}) arc (-15:-105:1.55)
    node[pos=0.5, below, font=\footnotesize] {$B$};
\draw[thick, violet] ({1.55*cos(15)},{1.55*sin(15)})
    .. controls (0.7,0.35) and (0.7,-0.35) .. ({1.55*cos(-15)},{1.55*sin(-15)});
\draw[thick, violet] ({1.55*cos(105)},{1.55*sin(105)})
    .. controls (-0.7,0.35) and (-0.7,-0.35) .. ({1.55*cos(-105)},{1.55*sin(-105)});
\draw[very thick, orange] (-0.9,0)
    .. controls (-0.35,0.55) and (0.35,0.55) .. (0.9,0);
\node[orange, font=\footnotesize] at (0,0.66) {$\Sigma_{A{:}B}$};
\draw[very thick, orange, dashed] (0.9,0)
    .. controls (0.35,-0.55) and (-0.35,-0.55) .. (-0.9,0);
\node[font=\footnotesize] at (1.15,-0.44) {$\Gamma_p$};
\draw[->, thin] (1.00,-0.44) -- (0.62,-0.48);
\end{tikzpicture}
\caption{Reflected-replica geometry (schematic projection). Two disjoint
boundary intervals $A$ and $B$ (blue) share joint RT arcs (violet) that
delimit the entanglement wedge of $A\cup B$. The entanglement-wedge cross
section $\Sigma_{A{:}B}$ (orange, solid) is doubled by canonical purification
into a mirror copy (orange, dashed); their reflection gluing forms the closed
contour $\Gamma_p$. In the linked filling shown, $\Gamma_p$ encloses the
projection of the closed background vortex-sector line $\cV_n$ (red), so
$\Lk_{\fmark}(\Gamma_p,\cV_n)=1$. More generally the signed linking
$\nu_{\fmark}$ is topological data of the fixed marked filling, in the sense
of Remark~\ref{rem:linking}, and enters the topological-limit phase of
$\cW_p^{(\fmark)}(n)$.}
\label{fig:linking}
\end{figure}
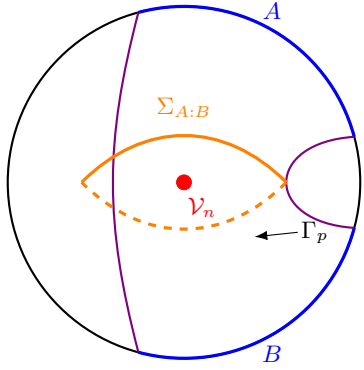

\subsection{Mutual monodromy from intrinsic spin}
\label{sec:monodromy}

For Abelian anyons the mutual monodromy follows from the topological spins.
Fusion of vortex sectors adds flux, $\Phi_{p+n}=\Phi_p+\Phi_n$, so
Eq.~\eqref{eq:spin} gives
\begin{equation}
2\pi\left(s_{p+n}-s_p-s_n\right)
=-\frac{\kappa}{2}\left[(\Phi_p+\Phi_n)^2-\Phi_p^2-\Phi_n^2\right]
=-\kappa\Phi_p\Phi_n,
\label{eq:spindiff}
\end{equation}
and therefore
\begin{equation}
M_{p,n}
=\exp\left\{2\pi i\left(s_{p+n}-s_p-s_n\right)\right\}
=\exp\left(-i\kappa\Phi_p\Phi_n\right).
\label{eq:monodromy}
\end{equation}

For CSH vortices the naive electromagnetic Aharonov-Bohm argument, which
couples only the electric charge $Q_p$ of the probe to the magnetic flux
$\Phi_n$ of the background, returns a phase equal to the inverse of the
mutual statistics inferred from the vortex spin. Kim and
Lee~\cite{KimLee1994} traced the discrepancy to a Magnus interaction between
the vortex and the condensate current, and showed that combining it with the
electromagnetic Aharonov-Bohm term into a single dual electromagnetic
interaction reproduces the spin-statistics-consistent phase.
Equation~\eqref{eq:monodromy}, derived from the intrinsic vortex spin, is the
convention-independent way to state that result.

\subsection{Linking phase and vortex-antivortex contrast}
\label{sec:contrast}

If the signed linking number of $\Gamma_p$ and $\cV_n$ in the fixed marked
reflected replica is $\nu_{\fmark}=\Lk_{\fmark}(\Gamma_p,\cV_n)$, then at
separation $d$ large compared with the vortex width $\xi=m^{-1}$,
\begin{equation}
\cW_p^{(\fmark)}(n)=M_{p,n}^{\nu_{\fmark}}\left[1+O(e^{-d/\xi})\right],
\label{eq:W}
\end{equation}
so that the normalized combination in Eq.~\eqref{eq:normalized} retains only
the connected linking contribution. It is useful to keep the geometric and
topological outputs together without conflating them into a new entropy.
Define the purification-holonomy data
\begin{equation}
\cP_p^{(\fmark)}(n;A{:}B)\equiv
\left(\EPH(A{:}B),\,\cW_p^{(\fmark)}(n)\right).
\label{eq:purificationholonomypair}
\end{equation}
In a connected linked saddle and in the topological limit,
\begin{equation}
\cP_p^{(\fmark)}(\pm n;A{:}B)=
\left(\frac{E_W(A{:}B)}{4G},\,
M_{p,n}^{\pm\nu_{\fmark}}\left[1+O(e^{-d/\xi})\right]\right).
\label{eq:pairseparation}
\end{equation}
The first entry is the common minimal purification cost, whereas the second
entry is inverted by charge conjugation. This ordered pair is not proposed as
a replacement for $E_p$ or $S_R$; it records the minimal geometric
purification data together with the defect phase that those entropies
discard.

Comparing charge-conjugate sectors gives the profile-independent
vortex-antivortex contrast
\begin{equation}
\boxed{\begin{aligned}
\cX_p^{(\fmark)}(n)
&\equiv\frac{\cW_p^{(\fmark)}(n)}{\cW_p^{(\fmark)}(-n)}\\
&=\exp\left(-2i\nu_{\fmark}\kappa\Phi_p\Phi_n\right)
\left[1+O(e^{-d/\xi})\right].
\end{aligned}}
\label{eq:central}
\end{equation}
Equations \eqref{eq:WilsonAlgebraBlind}, \eqref{eq:purificationblind},
\eqref{eq:pairseparation} and \eqref{eq:central} are the central formulas of
the construction. They may be summarized as the sector-selective holonomy law
\begin{equation}
\boxed{\cG_{\rm grav}(n)=\cG_{\rm grav}(-n),
\qquad
M_{p,-n}=M_{p,n}^{-1}.}
\label{eq:holonomyseparation}
\end{equation}
The first equality is exact for the classical gravitational sector, and the
second is the exact Abelian matter monodromy. Geometry is therefore blind to
the sign of the anyonic hair, while a linked reflected phase is orientation
sensitive in the topological limit.

The contrast distinguishes the two orientations whenever
$2\nu_{\fmark}\kappa\Phi_p\Phi_n\notin2\pi\mathbb Z$. At exceptional probe
charges or levels for which the contrast is unity, the chosen interferometer
aliases the two sectors and another allowed probe may be required. When the
matter Chern-Simons level is quantized this condition is not merely a
caveat: Proposition~\ref{prop:aliasing} below shows that no admissible probe
resolves the two orientations if and only if the sector is self-conjugate,
in which case there is no doublet to resolve. The
linking class $\nu_{\fmark}$ is topological data of the chosen marked replica
filling. The derivation does not assume that a connected EWCS automatically
implies unit linking: a linked connected saddle has $\nu_{\fmark}\neq0$,
whereas an unlinked filling has $\nu_{\fmark}=0$ and no orientation signal.

The superscript $(\fmark)$ is essential. Equation~\eqref{eq:central} is an
exact topological-limit contrast within one marked saddle, with identical
edge and framing data in numerator and denominator. A modular transformation
that changes the contractible cycle can map the calculation to a different
bulk filling, and no sum over modular images is evaluated here.

Representative topological-limit phases $\arg\cW_p^{(\fmark)}(n)$,
$\arg\cW_p^{(\fmark)}(-n)$ and $\arg\cX_p^{(\fmark)}(n)$ are shown in
Fig.~\ref{fig:main}(c) for $\vartheta_{11}\equiv-\kappa(2\pi/e)^2=2\pi/5$ and
$\nu_{\fmark}=1$. The contrast phase is twice the single-orientation phase
modulo $2\pi$, as follows directly from Eq.~\eqref{eq:central}.

\subsection{What the contrast does and does not measure}
\label{sec:whatitmeasures}

Two qualifications delimit the claim, and we state them explicitly because
each is a natural objection.

First, ${\sf C}$ is a symmetry of the action~\eqref{eq:action}, not an
explicit breaking. Conjugating the probe together with the background leaves
the monodromy invariant,
\begin{equation}
M_{-p,-n}=M_{p,n},
\qquad\text{whereas}\qquad
M_{p,-n}=M_{p,n}^{-1}.
\label{eq:relativeorientation}
\end{equation}
Equation~\eqref{eq:central} therefore resolves the \emph{relative}
orientation of the probe sector $p$ and the background sector $n$, at a fixed
external probe convention. Which member of the doublet is labelled $n$ and
which $-n$ is a convention, fixed once and for all by the choice of probe.
This is the standard situation for a discrete symmetry that is not gauged:
the two members are distinct states of the same Hilbert space, separated by
any operator odd under ${\sf C}$, while the absolute labelling carries no
invariant meaning. The ordered pair
$\cP_p^{(\fmark)}$ of Eq.~\eqref{eq:purificationholonomypair} is accordingly
${\sf C}$-covariant rather than ${\sf C}$-invariant, with the first entry
inert and the second entry inverted. Section~\ref{sec:noglobal} shows that
this is not an accident of the model: an exact ${\sf C}$ acting only on the
gravitational sector would be a bulk global symmetry of the kind excluded in
Refs.~\cite{HarlowOoguri2019,HarlowOoguri2021}.

Second, Eq.~\eqref{eq:central} is not the only matter-sector observable
capable of distinguishing the two orientations, and is not claimed to be.
Since the Chern-Simons Gauss law gives $Q_n=\kappa\Phi_n$ in
Eq.~\eqref{eq:gauss}, a direct measurement of charge or flux is already odd
under $n\to-n$. The purpose of the reflected construction is more specific.
It isolates a normalized mutual-monodromy phase \emph{within} a geometric
purification protocol, after common gravitational dressing, one-line factors
and local core contributions have been removed, and it does so through the
same cross section that computes $\EPH$ and one half of $S_R$. The
contribution is therefore the embedding of a topological order parameter in
the entanglement-wedge construction, not the discovery of a
${\sf C}$-odd quantity.

\section{Quantum error correction and the necessity of the matter line}
\label{sec:qec}

Sections~\ref{sec:cthm} to~\ref{sec:interferometer} are self-contained. This
section recasts them in the operator-algebra language of holographic quantum
error correction, and shows that within that language the matter line is not
an optional refinement of the geometric protocol but a required object.

\subsection{When the doublet exists: an exact aliasing criterion}
\label{sec:aliascriterion}

Before asking what can read the orientation, we fix when there is an
orientation to read. The nonaliasing condition of Sec.~\ref{sec:contrast}
becomes sharp once the matter Chern-Simons level is quantized. Suppose the
unit-vorticity monodromy angle of Fig.~\ref{fig:main}(c) satisfies
\begin{equation}
\vartheta\equiv-\kappa\left(\frac{2\pi}{e}\right)^2=\frac{2\pi}{N},
\qquad N\in\mathbb Z_{>0},
\label{eq:levelquantization}
\end{equation}
which for $N=5$ is the value used in that figure. The vortex sectors then
close into the Abelian anyon theory $\mathbb Z_N$, the label $n$ is defined
modulo $N$, and Eq.~\eqref{eq:monodromy} becomes
\begin{equation}
M_{p,n}=\exp\left(\frac{2\pi i\,pn}{N}\right).
\label{eq:ZNmonodromy}
\end{equation}

\begin{proposition}[Aliasing occurs exactly for self-conjugate sectors]
\label{prop:aliasing}
Assume Eq.~\eqref{eq:levelquantization} and let $\nu_{\fmark}=\pm1$. The
following are equivalent.
\begin{itemize}
\item[\rm(i)] $\cX_p^{(\fmark)}(n)=1$ in the topological limit for every
probe sector $p$.
\item[\rm(ii)] $2n\equiv0 \bmod N$.
\item[\rm(iii)] The sector is self-conjugate, $n\equiv-n \bmod N$.
\end{itemize}
For general $\nu_{\fmark}$, condition {\rm(i)} holds if and only if
$2\nu_{\fmark}n\equiv0\bmod N$.
\end{proposition}

\begin{proof}
By Eqs.~\eqref{eq:central} and~\eqref{eq:ZNmonodromy},
\begin{equation}
\cX_p^{(\fmark)}(n)=M_{p,n}^{2\nu_{\fmark}}
=\exp\left(\frac{4\pi i\,\nu_{\fmark}pn}{N}\right).
\label{eq:ZNcontrast}
\end{equation}
This equals unity for every integer $p$ if and only if
$2\nu_{\fmark}n\equiv0\bmod N$, which for $\nu_{\fmark}=\pm1$ is
$2n\equiv0\bmod N$, and that is the statement $n\equiv-n\bmod N$.
\end{proof}

The criterion is self-consistent in a way that the earlier caveat could not
express. The interferometer fails to separate the two orientations precisely
when the two orientations are the same anyon sector, in which case there is
no doublet and nothing has been missed. Whenever the sector is not
self-conjugate, some probe resolves it. This is the charge-conjugation
instance of nondegeneracy of the braiding form: in a modular Abelian theory
the only transparent sector is the vacuum~\cite{Kitaev2006}, and the relevant
transparency condition here is that of the fused sector $2n$. All statements
below assume a sector that is not self-conjugate.

\subsection{The doublet as an exact bulk logical qubit}
\label{sec:logicalqubit}

In the error-correcting formulation of AdS/CFT, bulk effective field theory
operators are logical operators on a code subspace of the boundary Hilbert
space, and a bulk operator is representable on a boundary region $R$ when it
lies in the entanglement wedge of
$R$~\cite{AlmheiriDongHarlow2015,JafferisEtAl2016,DongHarlowWall2016,JahnEisert2021}.
The relevant technical statement is the theorem of operator-algebra quantum
error correction~\cite{BenyKempfKribs2007,AlmheiriDongHarlow2015}: an
operator $O$ acting within a code subspace $\cC$ is representable on the
complement of an erased region $E$ if and only if
$\langle\tilde i|[O,X_E]|\tilde j\rangle=0$ for every $X_E$ supported on $E$
and every pair of code states. Correctability is a property of a subalgebra,
not of the full algebra of operators on $\cC$.

The charge-conjugate doublet supplies an unusually clean instance. Fix an
admissible branch of winding $n$ in the sense of
Definition~\ref{def:admissible}, let $|\Psi_{\pm n}\rangle$ denote the two
members of the pair produced by Theorem~\ref{thm:pairing}, and define the
two-dimensional code subspace
\begin{equation}
\cC_n\equiv\mathrm{span}\left\{|\Psi_{n}\rangle,\,|\Psi_{-n}\rangle\right\}.
\label{eq:codesubspace}
\end{equation}

\begin{proposition}[Gravitational blindness as a code property]
\label{prop:logical}
Let $\Alg_{\rm grav}$ denote the algebra generated by the common
gravitational data $\mathfrak D_{\rm grav}$ of
Eq.~\eqref{eq:gravdataexhaustion} together with all gravitational Wilson
amplitudes carrying matched dressing data in the sense of
Definition~\ref{def:dressing}. Let the sector $n$ be non-self-conjugate in
the sense of Proposition~\ref{prop:aliasing}. Then at leading semiclassical
order every $G\in\Alg_{\rm grav}$ obeys
\begin{equation}
G\big|_{\cC_n}=g\,\mathbf 1_{\cC_n},
\qquad g\in\mathbb C.
\label{eq:logicalqubit}
\end{equation}
\end{proposition}

\begin{proof}
The diagonal matrix elements coincide by Theorem~\ref{thm:holsep} and
Corollaries~\ref{cor:banados} to~\ref{cor:geom}, which give $G[n]=G[-n]$ for
every such observable. For the off-diagonal elements, note that the vortex
label is a superselection charge. By the Chern-Simons Gauss law of
Eq.~\eqref{eq:gauss} the sector is detected at the boundary by braiding with
a probe line, through the monodromy of Eq.~\eqref{eq:ZNmonodromy}, and by
Proposition~\ref{prop:aliasing} the two sectors are distinct. Every element
of $\Alg_{\rm grav}$ is built from the metric and the connections
$\cA^{(\pm)}$ of Eq.~\eqref{eq:gravconnections}, which are neutral under the
matter gauge group, so no such element carries vortex charge and
$\langle\Psi_{n}|G|\Psi_{-n}\rangle=0$. Hence $G$ restricted to $\cC_n$ is
diagonal with equal entries, which is Eq.~\eqref{eq:logicalqubit}.
\end{proof}

Equation~\eqref{eq:logicalqubit} states that $\cC_n$ is a one-qubit code on
which the entire classical gravitational algebra acts trivially. The relative
orientation of Sec.~\ref{sec:whatitmeasures} is a logical degree of freedom,
and Eq.~\eqref{eq:central} exhibits a logical operator for it. The structure
has a precedent in holographic quantum error correction and tensor-network
codes~\cite{AlmheiriDongHarlow2015,PastawskiEtAl2015}. Theorem~\ref{thm:replicablind}
adds that every neutral fixed-filling replica observable has equal diagonal
matrix elements on the two code states; with the same superselection argument
used in Proposition~\ref{prop:logical}, the corresponding neutral replica
algebra also acts proportionally to the identity.

\begin{remark}[Exact for the neutral algebra, resolved in the charged one]
\label{rem:nonpert}
Proposition~\ref{prop:logical} is an operator statement for the leading
classical gravitational algebra, and Theorem~\ref{thm:replicablind} is a
fixed-filling path-integral statement for the larger charge-conjugation-even
replica algebra. Both are exact with respect to those neutral observables
under the stated assumptions. The complementary statement is equally exact:
by Corollary~\ref{cor:chargedblind} the charge-odd sector separates the two
code states already at first order in the conjugate angle, with the universal
value of Eq.~\eqref{eq:oddcumulantvalue}. Charged matter lines and boundary
defect operators therefore lie outside $\Alg_{\rm grav}$ by design rather than
by omission, which is precisely what makes the pair a code with a readable
logical operator rather than an inaccessible degeneracy.
\end{remark}

\subsection{No bulk global symmetries and the necessity of the matter line}
\label{sec:noglobal}

The gravitational blindness proved in Sec.~\ref{sec:cthm} is, in a precise
sense, the diagnostic signature of a bulk global symmetry. Harlow and Ooguri
define a global symmetry with group $G$ in asymptotically AdS quantum gravity
by requiring, among other conditions, that for every nontrivial normal
subgroup $H\subset G$ there exist two gravitationally dressed operators that
transform in the same representation of the asymptotic conformal group but in
different representations of $H$~\cite{HarlowOoguri2021}.
Theorem~\ref{thm:holsep} and Corollary~\ref{cor:banados} realize this
condition explicitly for the group generated by ${\sf C}$: the two members of
the doublet carry identical gravitational dressing and identical Ba\~nados
functions, hence identical asymptotic conformal data by
Eq.~\eqref{eq:boundaryblind}, while the matter monodromy is inverted by
Eq.~\eqref{eq:holonomyseparation}.

Harlow and Ooguri prove that no such global symmetry can exist in a quantum
gravity theory dual to a boundary conformal field
theory~\cite{HarlowOoguri2019,HarlowOoguri2021}. Their mechanism is the
present theorem run in the opposite direction. Decomposing the boundary into
regions and using entanglement wedge
reconstruction~\cite{DongHarlowWall2016,JafferisEtAl2016}, the symmetry
charge is expressible through operators whose entanglement wedges reach only
the gravitational dressing of the charged object; if two objects have
identical dressing they cannot then transform differently. Our
Theorem~\ref{thm:holsep} establishes exactly the hypothesis on which that
argument turns, namely that the dressing is identical, and
Theorem~\ref{thm:replicablind} strengthens it, since the agreement holds not
only classically but for every neutral fixed-filling replica observable.

The conclusion is not that the doublet is inconsistent. It is that the
distinguishing observable cannot lie in the neutral gravitational sector, and
that the theory must supply it elsewhere. The present model does so
explicitly, and Corollary~\ref{cor:chargedblind} says where. The matter
Chern-Simons field of Eq.~\eqref{eq:action} is a long-range bulk gauge
symmetry, so an operator creating a gauge-charged object in the interior is
attached to the boundary by a line that enters the entanglement wedge of some
boundary region; that line is $\cV_n$ of Sec.~\ref{sec:replica}. Its entropic
counterpart is the charge-odd part of the replica partition function, which by
Eq.~\eqref{eq:evenoddsplit} is exactly odd in the winding and by
Eq.~\eqref{eq:oddcumulantvalue} is nonzero in every nontrivial sector. The
interferometer of Eq.~\eqref{eq:central} is therefore not an optional addition
to the geometric protocol. It is the topological-limit representative of an
observable that the structure of the theory already contains.

Two further correspondences are worth recording. First, the completeness of
gauge representations, conjectured
earlier~\cite{Polchinski2004,BanksSeiberg2011} and established within
AdS/CFT by Harlow and Ooguri, asserts that dynamical objects exist in every
irreducible representation of a long-range bulk gauge group. In the Abelian
Chern-Simons setting this says that every probe sector $p$ in
Eq.~\eqref{eq:ZNcontrast} is populated, and Proposition~\ref{prop:aliasing}
then guarantees a distinguishing probe whenever the doublet is genuine.
Completeness and nonaliasing are the same requirement, stated once in the
language of quantum gravity and once in the language of anyon braiding.

Second, the argument establishing completeness uses a Wilson line anchored on
the two asymptotic boundaries of the maximally extended two-sided black hole
and threading the wormhole between
them~\cite{HarlowOoguri2021}. The two-sided protocol of
Sec.~\ref{sec:twosidedpurification} places $A$ and $B$ on opposite
boundaries and supports the probe contour on the cross section of a
wormhole-spanning entanglement wedge, so the geometry in which completeness
is established is the geometry in which our gate operates. We do not claim an
identification of the two line operators, only that the two constructions
live in the same background for the same reason: a connected wedge is what
allows a charged line to be anchored at the boundary at all.

\subsection{Wedge connectivity as a geometric access threshold}
\label{sec:erasure}

The disconnection transition of Sec.~\ref{sec:twosidedpurification} has a
natural reading in the same framework, and it answers a question posed in the
original error-correction analysis. Almheiri, Dong and Harlow consider an
AdS$_3$ boundary region consisting of two disjoint intervals, note that the
minimal surface jumps discontinuously as the intervals grow, and observe that
a bulk operator at the centre is reconstructible on that region if and only
if the induced division of the boundary constitutes a quantum secret sharing
scheme~\cite{AlmheiriDongHarlow2015}. They were not able to settle the
question for a local bulk field.

Section~\ref{sec:purificationgate} supplies a solvable version of the same
configuration in which the encoded object is topological rather than local.
For $0\leq T_0<T_*$ the entanglement wedge of $A\cup B$ is connected, the
doubled cross section $\Gamma_p$ of Fig.~\ref{fig:linking} exists with
$\nu_{\fmark}\neq0$, and the quantized phase of
Eq.~\eqref{eq:timeindependentcontrast} is available. At $T_0\geq T_*$ the
wedge disconnects and the cross-section-supported contour ceases to exist, as
recorded in Eq.~\eqref{eq:physicalEPpiecewise}. Wedge connectivity is
therefore the threshold separating availability from nonavailability of this
specific cross-section-supported logical readout. In that restricted sense,
$T_*$ of Eq.~\eqref{eq:Tstar} is a geometric erasure threshold for the line
protocol. In holographic tensor-network codes the analogous statement is that
a bulk logical degree of freedom placed deep in the network is recoverable
only from boundary regions whose wedge reaches
it~\cite{PastawskiEtAl2015,JahnEisert2021}.

The scope of this reading is worth stating once. What
Sec.~\ref{sec:twosidedpurification} establishes is the availability of one
specific cross-section-supported protocol, a geometric fact following from
Eqs.~\eqref{eq:Tstar} and~\eqref{eq:physicalEPpiecewise}, which the
error-correction framework then interprets. Closing the channel removes this
observer's access to the phase rather than setting the phase to zero, and by
Corollary~\ref{cor:chargedblind} the orientation remains encoded in the
charge-odd sector of the same replica even where the leading geometric channel
has closed.

With this reading, the universal jump height of Eq.~\eqref{eq:universaljump}
acquires a code-theoretic gloss. The quantity
\begin{equation}
\EPH(T_*^-)=\frac{c}{3}\log\left(1+\sqrt2\right)
\label{eq:jumpcode}
\end{equation}
is the geometric purification cost at the point where the specified line
access closes. It is
independent of the cutoff, the temperature and Newton's constant, and is
fixed by the central charge alone, yet by Proposition~\ref{prop:logical} it
carries no information whatever about the logical state it gates.

\section{Discussion and conclusion}
\label{sec:discussion}

The two Chern-Simons structures of the problem encode complementary
information. The $\mathrm{SL}(2,\mathbb R)_+\times\mathrm{SL}(2,\mathbb R)_-$
gravitational connections encode the geometry, the BTZ conjugacy classes and
the Brown-Henneaux charges. The Abelian CSH sector encodes charge-flux
attachment, fractional spin and mutual braiding. Charge conjugation leaves
the gravitational connections unchanged and reverses the matter sector. This
yields a stronger result than metric blindness alone: all closed
gravitational holonomy characters, and all open or network Wilson amplitudes
evaluated on the same contours with identical representation and dressing
data, are exactly even under $n\to-n$.

The charge-conjugate pair therefore forms a gravitationally degenerate
topological doublet. Its members have the same metric, the same BTZ holonomy
classes, the same mass and angular momentum, the same semiclassical Virasoro
data and the same leading geometric entanglement observables. They are not
the same physical state. A matter topological line sees inverse mutual
monodromy, and the normalized reflected ratio converts this difference into
the phase
\begin{equation}
\cX_p^{(\fmark)}(n)=\exp\left[-2i\nu_{\fmark}\kappa\Phi_p\Phi_n\right],
\label{eq:discussioncontrast}
\end{equation}
up to exponentially small finite-core corrections. The result is operational
rather than merely classificatory: it specifies the entire class of probes
that must fail, together with a line operator that succeeds within the same
purification geometry, in the precise sense delimited in
Sec.~\ref{sec:whatitmeasures}.

This formulation also clarifies the boundary interpretation. Equality of the
Brown-Henneaux charges and of the stress-tensor one-point data means that the
charge-conjugate sectors are indistinguishable within the leading
semiclassical Virasoro channel. Their distinction resides in a
topological-defect sector that is not generated by the stress tensor. The
reflected interferometer is therefore a defect-sensitive completion of
geometric entanglement, not a correction to the geodesic length itself. In
this sense the missing hair is not hidden in a more accurate metric
reconstruction. It lies outside the neutral gravitational-replica observable
algebra.

The fixed-filling replica theorem makes this separation robust across the
first quantum correction and, at the level of the change of variables, across
the complete uncharged replica path integral. The paired saddles have
isospectral gauge-fixed fluctuation operators, their FLM bulk entropies and
local renormalization terms agree, and the generalized-entropy functionals
coincide pointwise on the common geometry, so quantum extremization cannot
select different surfaces for $n$ and $-n$.

Corollary~\ref{cor:chargedblind} completes the accounting. The same change of
variables that closes the neutral sector opens the charged one exactly. The
charge-odd part of the replica free energy is odd in the winding, the
symmetry-resolved entanglement spectrum obeys $\cZ_q(\cQ;n)=\cZ_q(-\cQ;-n)$,
and the first charged cumulant reverses with the winding with the topological
value of Eq.~\eqref{eq:oddcumulantvalue}. Summing over charge sectors
reconstructs the blind total entropy, so the orientation is stored in the
charge asymmetry of the entanglement spectrum and nowhere else. The linked
line amplitude of Eq.~\eqref{eq:discussioncontrast} is the topological-limit
representative of that same odd datum, which places the geometric no-go and
the matter-sector readout inside one identity.

The entanglement-of-purification perspective sharpens this statement. Even a
minimization over geometric purifications, together with the reflected
entropy, the mutual information and their known inequality structure, returns
only the common first entry of Eq.~\eqref{eq:pairseparation}. The exact
one-sided BTZ plateau and the two-sided disconnection transition show that
this geometric entry can have rich and nonanalytic behaviour without
acquiring any orientation sensitivity. Equation~\eqref{eq:universaljump}
makes the point quantitative: the leading geometric gate closes with a jump
whose height is fixed by the central charge alone and is independent of the
cutoff, the temperature and Newton's constant, yet that number carries no
information about the vortex sector. Any order-$G^0$ mutual information
remaining after the classical disconnection is charge-conjugation even, and is
therefore orientation blind as well. Conversely, the matter phase can remain perfectly
quantized while the cross-section length evolves. The proposed readout is
therefore a modulus-phase separation in operational terms: purification
geometry sets the correlation cost and opens or closes the channel, while
matter holonomy carries the sign information.

The linearized BPS-core calculation adds a quantitative but conditional
complement to the exact symmetry theorem. If a regular stationary localized
branch exists and its $O(G)$ charge shifts are controlled by the flat-space
BPS data, Proposition~\ref{prop:exteriorimprint} gives
$\delta\widehat M=\gamma|n|$ and
$|\delta\widehat j|=\gamma n^2/(mL)$ for a nonrotating seed. These shifts
are again even in $n$, as they must be. They should be interpreted as
asymptotic charge matching, not as a construction of the missing black-hole
solution.

Within the additional rotating-BTZ matching assumption of
Sec.~\ref{sec:band}, the two charge scalings generate a finite continuous
candidate envelope in $|n|$. Its algebraic massless-seed upper edge is
$|n|=mL$, while its continuous lower edge approaches $M_c/\gamma$ at large
$\lambda$. The physical candidate set is smaller: $n$ is discrete and the
BPS source must be localized behind the horizon. Indeed, the numerical
$R_{1/2}\simeq2|n|/m$ scaling shows that the algebraic upper edge is not
localized in the plotted range $\lambda\leq6$. The robust conclusion is
therefore the existence of a computable conditional envelope, not a theorem
that every point in that envelope is realized by a gravitating vortex. The
cutoff-independent jump and Markov-gap values derived earlier remain exact
statements of the nonrotating BTZ benchmark; their numerical values are not
extended here to the induced rotating case.

The error-correcting reading of Sec.~\ref{sec:qec} sharpens the claim once
more. The doublet is a one-qubit code for the neutral gravitational-replica
algebra, whose logical information no classical gravitational observable can
read, and whose distinguishing logical operator is a
topological line rather than a local bulk field. That such an operator must
exist is not a modelling choice. A charge-conjugation symmetry acting
faithfully on matter while acting trivially on the entire gravitational
sector is, by Theorem~\ref{thm:holsep}, precisely the configuration that the
no-global-symmetries theorem excludes as an exact bulk global
symmetry~\cite{HarlowOoguri2019,HarlowOoguri2021}. The resolution offered
there, that the symmetry be a long-range gauge symmetry whose charged objects
carry boundary-anchored lines, is realized here explicitly and computably.
The completeness of gauge charges required by the same analysis is, in this
model, the statement that a distinguishing probe exists whenever the two
orientations are genuinely distinct sectors, which is
Proposition~\ref{prop:aliasing}.

The code-theoretic reading transfers directly to finite-dimensional
holographic models, where bulk logical operators and their recovery regions
are realized explicitly~\cite{PastawskiEtAl2015,JahnEisert2021} and where the
FLM split between an area term and a bulk entropy has been measured
directly~\cite{BiswasEtAl2026}. Our results predict a sharp signature for such
models. Neutral entropy channels remain exactly degenerate across a
charge-conjugate logical pair, so the measured area term is rigid under
reversal of the logical label, while the charge-resolved spectrum separates
the two members at first order in the conjugate angle.

Three broader implications follow. First, the theorem exhausts the classical
Einstein-gravity channels available under the stated hypotheses: local
curvature, global parallel transport and Brown-Henneaux boundary data are all
identical. The conclusion follows from equality of the full geometric fields
and does not rely on treating closed holonomies as a complete coordinate
system in the sourced core. Second, the modulus-phase decomposition
generalizes to any bulk theory containing a gravitational topological sector
and an independent matter topological sector, with non-Abelian mutual
monodromy replacing Eq.~\eqref{eq:monodromy}. Third, the product-holonomy
viewpoint provides a natural language for quantum corrections, and fixes them
in both sectors at once. Corrections constructed from the
charge-conjugation-even sector cannot split the pair, including the FLM
one-loop contribution and quantum extremization, while the charge-odd sector
splits it at first order with the universal value of
Eq.~\eqref{eq:oddcumulantvalue}. Fourth, a separate linearized charge-matching analysis turns the even BPS data
into the conditional continuous envelope of Eq.~\eqref{eq:bandcondition};
this estimate requires the stationary-branch, rotating-matching and
localization assumptions stated in Secs.~\ref{sec:exterior} and~\ref{sec:band}.

The charge-conjugation theorem itself is not restricted to stationary
solutions. The BTZ purification application remains conditional on the
existence of the stationary asymptotically BTZ CSH branch specified in
Appendix~\ref{sec:scope}. Existing backreacted solutions establish the
compatibility of vortex sectors with gravity, but not the desired regular
positive-mass nonextremal CSH black hole. Within one fixed marked replica
filling, the gravitational equality is nonperturbative in Newton's constant
and the matter-sector phase is universal in the topological limit.



\bibliographystyle{JHEP}
\bibliography{cs_vortex_holonomy_updated}

\appendix

\section{Gravitational Chern-Simons conventions and BTZ holonomies}
\label{sec:gravcsapp}

With the conventions of Eq.~\eqref{eq:gravconnections}, the field strengths
of the gravitational connections are
\begin{equation}
F^{(\pm)a}
=
R^a+\frac{1}{2L^2}\epsilon^a{}_{bc}e^b\wedge e^c
\pm\frac{1}{L}T^a,
\label{eq:gravfieldstrength}
\end{equation}
where
\begin{equation}
T^a=\dd e^a+\epsilon^a{}_{bc}\omega^b\wedge e^c,
\qquad
R^a=\dd\omega^a+\frac{1}{2}\epsilon^a{}_{bc}\omega^b\wedge\omega^c.
\label{eq:torsioncurvature}
\end{equation}
For torsion-free vacuum Einstein gravity with $\Lambda=-1/L^2$, both field
strengths vanish. Matter sources make the connections nonflat in the core,
but Eq.~\eqref{eq:Tparity} leaves the triad and spin connection unchanged
under charge conjugation. Equality of the closed holonomy characters and of
identically dressed gravitational Wilson amplitudes therefore requires
neither flatness nor path independence, as emphasized in
Remark~\ref{rem:noflatness}.

For the BTZ asymptotic connection, representatives of the angular-cycle
holonomies can be chosen as
\begin{equation}
\rho_L=
\begin{pmatrix}
e^{\pi(r_+-r_-)/L}&0\\
0&e^{-\pi(r_+-r_-)/L}
\end{pmatrix},
\qquad
\rho_R=
\begin{pmatrix}
e^{\pi(r_++r_-)/L}&0\\
0&e^{-\pi(r_++r_-)/L}
\end{pmatrix}.
\label{eq:BTZholmatrices}
\end{equation}
Their conjugacy classes determine the two combinations $r_+-r_-$ and
$r_++r_-$, and hence the BTZ charges in Eq.~\eqref{eq:BTZcharges}. The
equality in Eq.~\eqref{eq:gravholonomyblind} therefore gives a
gauge-invariant statement of the common asymptotic geometry of the
charge-conjugate pair.

\section{Analytic holographic purification formulas in BTZ}
\label{sec:btzpurapp}

For completeness we collect the special BTZ formulas used in
Sec.~\ref{sec:purificationgate}, derive
Eqs.~\eqref{eq:universaljump} and~\eqref{eq:universalmarkov}, and record the rotating-BTZ matching diagnostic used conditionally in
Eq.~\eqref{eq:rotatinggate}.

\subsection{One-sided equal intervals}

In a one-sided nonrotating BTZ geometry, two equal intervals each slightly
smaller than half of the boundary circle have a horizon-threading
entanglement wedge. If their half-width is $\alpha$, the cross section can be
written as
\begin{equation}
\EPH(r_+,\alpha)=\frac{L}{2G}
\log\!\left[\coth\!\left(
\frac{r_+}{2L}\left(\frac{\pi}{2}-\alpha\right)\right)\right].
\label{eq:equalintervalBTZEP}
\end{equation}
As $\alpha\to\pi/2$ the argument of the hyperbolic cotangent tends to zero
and Eq.~\eqref{eq:equalintervalBTZEP} diverges logarithmically. Regulated at
$r=r_c$ it becomes
\begin{equation}
\EPH\!\left(r_+,\frac{\pi}{2}\right)
=\frac{L}{2G}\log\!\left(\frac{2r_c}{r_+}\right).
\label{eq:halfboundaryBTZEP}
\end{equation}
These are the radial-geodesic branches that generate the middle line of
Eq.~\eqref{eq:btzpurplateau}, and Eq.~\eqref{eq:halfboundaryBTZEP} coincides
with the plateau height there, as it must.

\subsection{Two-sided geometry and the universal jump}

For the two-sided geometry the connected HRT saddle gives
Eqs.~\eqref{eq:twosidedEP} and \eqref{eq:twosidedMI}. At $T_0=0$,
Eq.~\eqref{eq:twosidedEP} reduces to
\begin{equation}
\EPH(0)=\frac{L}{4G}\arccosh\!\left[\cosh\!\left(\frac{\pi L}{z_H}\right)\right]
=\frac{\pi r_+}{4G},
\label{eq:twosidedEPzero}
\end{equation}
that is, to one half of the bifurcation-circle circumference $2\pi r_+$
divided by $4G$.

At the transition time the argument of Eq.~\eqref{eq:twosidedEP} collapses to
a pure number. Writing $S\equiv\sinh(\pi L/2z_H)$, Eq.~\eqref{eq:Tstar} gives
$\cosh(T_*/z_H)=S$, so
\begin{equation}
\operatorname{sech}^{2}\!\left(\frac{T_*}{z_H}\right)=\frac{1}{S^{2}},
\qquad
\cosh\!\left(\frac{\pi L}{z_H}\right)-1=2S^{2},
\label{eq:jumpingredients}
\end{equation}
using the double-angle identity $\cosh 2y-1=2\sinh^2y$. The brace in
Eq.~\eqref{eq:twosidedEP} is therefore $1+2S^2/S^2=3$, independently of $L$,
$z_H$ and $G$. Since $\sinh y=1$ implies $\cosh 2y=1+2\sinh^2y=3$, one has
$\arccosh 3=2\arcsinh 1=2\log(1+\sqrt2)$, and Eq.~\eqref{eq:universaljump}
follows.

The associated Markov gap follows immediately. The connected two-sided mutual
information of Eq.~\eqref{eq:twosidedMI} vanishes at $T_*$ by construction,
while the reflected entropy of the connected saddle is $S_R=2\EPH$ at leading
order~\cite{DuttaFaulkner2021}. Hence
\begin{equation}
\Delta_M(T_*^-)=\left(2\EPH-I\right)\big|_{T_*^-}
=\frac{2c}{3}\log\left(1+\sqrt2\right)
=\frac{c}{3}\log\left(3+2\sqrt2\right),
\label{eq:markovderivation}
\end{equation}
using $(1+\sqrt2)^2=3+2\sqrt2$. Like Eq.~\eqref{eq:universaljump}, this value
is independent of the cutoff, the horizon radius and Newton's constant, and it
is common to the charge-conjugate pair.

\subsection{Rotating BTZ matching diagnostic}
\label{sec:rotatingplateau}

The nonrotating plateau criterion used in the main text follows from a static
reflection-symmetric geometry. Rotating BTZ is stationary rather than static,
so a covariant EWCS must in general be obtained from extremal rather than
constant-time minimal surfaces. We therefore separate two facts. First, the
standard equal-time interval entropy factorizes exactly into left- and
right-moving sectors. Second, a simple radial proper-length candidate can be
written in closed form. Their matching motivates the diagnostic used in
Sec.~\ref{sec:band}, but we do not claim here that these candidates exhaust the
globally minimal covariant EWCS for every rotating configuration.

For rotating BTZ,
$F=(r^2-r_+^2)(r^2-r_-^2)/(L^2r^2)$, and the fixed-$t$, fixed-$\varphi$ radial
proper length is
\begin{equation}
\ell_{\rm rad}
=\int_{r_+}^{r_c}\frac{\dd r}{\sqrt{F}}
=L\log\!\left[
\frac{\sqrt{r_c^2-r_+^2}+\sqrt{r_c^2-r_-^2}}
{\sqrt{r_+^2-r_-^2}}\right]
\longrightarrow
L\log\frac{2r_c}{\sqrt{r_+^2-r_-^2}} .
\label{eq:rotradial}
\end{equation}
The equal-time single-interval entropy at half-width $\alpha$ is
\cite{HRT2007,HubenyMaxfieldRangamaniTonni2013}
\begin{equation}
S(\alpha)=\frac{L}{4G}\log\!\left[
\frac{4r_c^2}{r_+^2-r_-^2}
\sinh\!\left(\frac{(r_+-r_-)\alpha}{L}\right)
\sinh\!\left(\frac{(r_++r_-)\alpha}{L}\right)\right].
\label{eq:rotinterval}
\end{equation}
Matching these regulated expressions cancels the cutoff and gives
\begin{equation}
\sinh\!\left(\frac{(r_+-r_-)\alpha_c}{L}\right)
\sinh\!\left(\frac{(r_++r_-)\alpha_c}{L}\right)=1 .
\label{eq:rotcrossover}
\end{equation}
If this matching continues to identify the relevant covariant cross-section
branch, the condition $\alpha_c\leq\pi/2$ becomes Eq.~\eqref{eq:rotatinggate}
through Eq.~\eqref{eq:dimensionlessBTZcharges}. Setting $r_-=0$ returns the
nonrotating criterion Eq.~\eqref{eq:plateaucondition} exactly.

Two scope statements are important. Reality of the left-moving factor imposes
$|\widehat j|\leq\widehat M$, but this nonextremality condition is logically
separate from the matching inequality itself. Moreover, the universal jump
and Markov-gap values in Eqs.~\eqref{eq:universaljump}
and~\eqref{eq:universalmarkov} were derived for $r_-=0$; no claim is made that
the same numerical heights survive for $r_-\neq0$.

Returning now to the nonrotating two-sided benchmark of
Sec.~\ref{sec:twosidedpurification}, the physical holographic purification
cost is obtained by comparing the connected saddle with the disconnected
one:
\begin{equation}
E_{{\rm P},{\rm phys}}^{\rm(h)}(T_0)=
\begin{cases}
\EPH(T_0),&0\leq T_0<T_*,\\
0,&T_0\geq T_*.
\end{cases}
\label{eq:physicalEPpiecewise}
\end{equation}
The discontinuity at $T_*$ is geometric and common to the charge-conjugate
pair, and its height is Eq.~\eqref{eq:universaljump}.
Equation~\eqref{eq:timeindependentcontrast} is instead conditional on a
linked filling of the connected purification saddle, and does not extend the
cross-section-supported protocol into the disconnected phase.

\section{Replica change of variables and one-loop determinant pairing}
\label{sec:replicaapp}

This appendix records the perturbative content of
Theorem~\ref{thm:replicablind}, Corollary~\ref{cor:flmblind} and
Corollary~\ref{cor:chargedblind}. The one-loop pairing can be checked directly
at the level of the gauge-fixed quadratic operators, which makes the
fixed-filling identities explicit order by order.

Collect all real fluctuation fields, including metric, scalar, gauge and ghost
components, into $\eta$. Around a charge-conjugate saddle pair
$\Psi_{-n}={\sf C}\Psi_n$, write
\begin{equation}
I_{E,\rm gf}[\Psi_n+\eta]
=I_{E,\rm gf}[\Psi_n]
+\frac12\langle\eta,\mathcal K_n\eta\rangle+O(\eta^3).
\label{eq:quadraticexpansion}
\end{equation}
The action of ${\sf C}$ on the fluctuation multiplet is a linear involution
$U_{\sf C}$ with $U_{\sf C}^2=\mathbf 1$, hence invertible with unit-modulus
Jacobian. Charge-conjugation invariance of the gauge-fixed action gives
\begin{equation}
\mathcal K_{-n}=U_{\sf C}\mathcal K_nU_{\sf C}^{-1}.
\label{eq:hessiansimilarity}
\end{equation}
Since Eq.~\eqref{eq:hessiansimilarity} is a similarity transformation, the
bosonic and ghost spectra are paired with their multiplicities. For any common
charge-conjugation-invariant regulator,
\begin{equation}
{\det}'\mathcal K^{\rm bos}_{-n}={\det}'\mathcal K^{\rm bos}_{n},
\qquad
{\det}'\mathcal K^{\rm gh}_{-n}={\det}'\mathcal K^{\rm gh}_{n},
\label{eq:detpairing}
\end{equation}
where the prime denotes the same treatment of zero modes in the paired
sectors. The one-loop effective action
\begin{equation}
\Gamma^{(1)}_{q,\fmark}
=\frac12\log{\det}'\mathcal K^{\rm bos}_{q,\fmark}
-\log{\det}'\mathcal K^{\rm gh}_{q,\fmark}
\label{eq:oneloopdef}
\end{equation}
is therefore the same in the two sectors. Local counterterms are built from
the common metric and charge-conjugation-even field combinations, so the
renormalized equality survives, which proves Eq.~\eqref{eq:oneloopblind}
explicitly.

For the bulk entropy term, let $\Alg(\Sigma_X)$ be the same gauge-invariant
operator algebra, including the same center or edge prescription, in the two
sectors. Charge conjugation induces an algebra automorphism $\alpha_{\sf C}$
and maps the reduced state according to
\begin{equation}
\rho_{\Sigma_X,-n}
=U_{\sf C,\Sigma_X}\rho_{\Sigma_X,n}U_{\sf C,\Sigma_X}^{\dagger}
\label{eq:reducedstatepairing}
\end{equation}
whenever a Hilbert-space representation is used; equivalently, the two
restricted algebraic states differ only by $\alpha_{\sf C}$. Their spectra and
therefore their algebraic entropies coincide,
\begin{equation}
S_{\rm bulk}^{\rm alg}(\Sigma_X;n)
=S_{\rm bulk}^{\rm alg}(\Sigma_X;-n).
\label{eq:bulkentropypairing}
\end{equation}
The area term is equal because the metric is common, and the FLM local terms
are equal by the same regulator and counterterm argument. This establishes
Eq.~\eqref{eq:genblind} for every $X$, not only at an extremum. Taking a first
variation then shows that the stationary sets coincide; comparing their
values shows that the minima and all degeneracies among competing quantum
extremal surfaces also coincide.

The complementary sector is fixed by the same automorphism. A charged
insertion is not an element of the even algebra, since $\alpha_{\sf C}$ maps a
probe $W_p$ to $W_{-p}$ and the charge $\cQ_A$ to $-\cQ_A$. Applying the
change of variables to the generating functional of
Eq.~\eqref{eq:chargedreplica} therefore gives
Eq.~\eqref{eq:chargedpairing} rather than
Eq.~\eqref{eq:replicablind}, so the charge-odd cumulants and the
symmetry-resolved spectrum reverse with the winding while the neutral
observables do not. The matter-sector contrast of Eq.~\eqref{eq:central} is
the topological-limit representative of that odd data.

\section{Scope of the BTZ statement}
\label{sec:scope}

This appendix states precisely which parts of the article are unconditional
and which presuppose a gravitating vortex branch that is not constructed
here.


The classical statement proved in Sec.~\ref{sec:cthm} is conditional only on
the existence of one admissible solution in the sense of
Definition~\ref{def:admissible}. The fixed-filling quantum extension of
Sec.~\ref{sec:replicablind} additionally requires the measure, gauge fixing,
regulator and algebraic edge prescription of
Definition~\ref{def:quantumadmissible}; it does not assume the missing BTZ
vortex branch beyond whatever saddle is being paired. If $(g_{\mu\nu},\phi,A_\mu)$ is such a
solution, then $(g_{\mu\nu},\phi^*,-A_\mu)$ is a second solution with the
same metric and opposite winding. This map requires neither stationarity nor
a uniqueness theorem, and it applies nonperturbatively in $G$ because it uses
only charge conjugation and invariant boundary data, not the Bogomolny
equations. Stationarity is needed only to interpret the common asymptotic
data as BTZ parameters and to use the explicit thermal purification formulas
of Sec.~\ref{sec:purificationgate}.


The distinction between the flat-space BPS calculation and a fully
backreacted gravitating Bogomolny system is important. London showed that a
specific eighth-order potential is required to reduce both the Einstein and
CSH equations to first order in his cylindrically symmetric
construction~\cite{London1995}. Cl\'ement's exact solutions likewise use a
gravitationally modified eighth-order potential and interpolate between two
constant-curvature regions~\cite{Clement1996}, and exact self-dual
gravitating CSH solitons are known in related
settings~\cite{ChungEtAl2001,KimKim2000}. Their asymptotic metrics are
locally of extreme-BTZ form, but the generic solutions have cylindrical
spatial topology, unquantized total flux and a negative BTZ mass parameter.
They therefore do not realize the positive-mass nonextremal thermal setup of
this article. Their narrower role is to show explicitly that a gravitating
CSH system can preserve the charge-conjugation pairing while supporting
nontrivial vortex structure.


Other theories establish complementary facts. The exact AdS$_3$ scalar vortex black holes of Cadoni, Pani and Serra possess inner and outer horizons and depend on the winding through $n^2$, but their scalar profile is singular at the origin and there is no local CSH gauge sector~\cite{CadoniPaniSerra2010}. Edery's Einstein-Maxwell-Higgs solutions
are regular and have positive mass, but remain
horizonless~\cite{Edery2020}. The horizon-piercing solutions of Ach\'ucarro,
Gregory and Kuijken demonstrate that a local gauge vortex can constitute
measurable black-hole hair, but in a four-dimensional asymptotically flat
Maxwell-Higgs setting~\cite{AchucarroGregoryKuijken1995}. A related
three-dimensional construction in the Nielsen-Olesen case was studied in
Ref.~\cite{Ghosh2021}. These results support the individual ingredients of
the proposed physical picture without providing the missing
Einstein-AdS$_3$-CSH branch.

An especially close neighbouring example is provided by Kim, Kim and Kimm,
who numerically constructed smooth, fully backreacted global $U(1)$ vortices
in AdS$_3$. As the magnitude of the negative cosmological constant is
lowered, their solutions pass from regular horizonless geometries through an
extremal configuration to a two-horizon charged black hole. Their asymptotic
metric, however, contains the logarithmic Goldstone contribution
\begin{equation}
 B(r)=|\Lambda|r^2-8\pi Gv^2n^2\ln(r/r_c)
      -8G\mathcal M+1+O(r^{-2}),
\label{eq:KKKasymptotics}
\end{equation}
and therefore belongs to a logarithmically deformed charged-BTZ class rather
than the Brown-Henneaux BTZ class assumed here. A duality transformation maps
the global-vortex winding to an electric Gauss-law charge, so the metric is
even in $n$ while the dual matter charge is odd in $n$. This gives an
independent gravitating-vortex realization of orientation-blind geometry and
sign-sensitive matter data, but it does not provide the localized
Einstein-CSH branch required for the present thermal purification
construction~\cite{KimKimKimm1997}.

No cited construction simultaneously provides a regular core, positive mass,
a nonextremal BTZ horizon and the CSH vortex sector used for the matter
monodromy. For the fixed sixth-order potential in Eq.~\eqref{eq:potential},
constructing such a branch remains a separate problem. The replica
identities of Sec.~\ref{sec:replicablind} and
Sec.~\ref{sec:chargedsector} are symmetry statements and are independent of
it.

Section~\ref{sec:exterior} makes one part of the residual gap quantitative.
At fixed winding, the flat-space self-dual source has the exponentially small
energy tail $\Delta_n(r_+)$ of Eq.~\eqref{eq:tailfraction}, so its integrated
charge matching becomes increasingly insensitive to the profile outside a
large putative horizon. This is not, however, a proof of horizon regularity.
For a stationary rotating horizon generated by
$\chi=\partial_t+\Omega_H\partial_\varphi$, the regular gauge-field
condition is formulated in terms of $\chi^\mu A_\mu=A_t+\Omega_HA_\varphi$
(and the scalar phase), not $A_t$ alone. The flat-space BPS relation in
Eq.~\eqref{eq:electricequalspotential} therefore does not impose the full
stationary horizon condition. The numerical values of $1-g(r_+)^2$ quoted by
the support code are useful electric-tail diagnostics in the nonrotating
limit, but they are not a gauge-invariant bound on the complete horizon
defect. The perturbative exterior analysis consequently assumes both small
$\Delta_n(r_+)$ and the existence of a regular stationary branch; it does not
construct that branch.


The flat-space BPS profiles are exponentially localized, so a weakly dressed
asymptotically AdS$_3$ solution would approach a BTZ metric at large radius
with corrections controlled by the matter tail. Since the tail is nonzero at
every finite radius, the metric need not become exactly BTZ outside a sharp
core. Establishing horizon regularity, global causal structure, positive mass and the nonextremality inequality $|J|<ML$ requires solving the full coupled boundary-value problem. An especially stringent numerical test would combine independent asymptotic-charge extraction with a matter-integral or constraint-based check, following the two-method strategy used for regular gravitating vortices~\cite{Edery2020, Ghosh2021}.


Accordingly, this article proves six statements. First, every classical
Einstein-gravity observable built from the common metric, the first-order
connections and matched boundary dressing data is exactly even under
$n\to-n$ (Theorems~\ref{thm:pairing} and~\ref{thm:holsep}). Second, all
classical Brown-Henneaux modes and, on a stationary BTZ branch, all BTZ
charges, horizons and leading geometric purification observables are
identical (Corollaries~\ref{cor:banados}, \ref{cor:btz}
and~\ref{cor:geom}). Third, for a quantum-admissible fixed filling, every
charge-conjugation-even replica partition function is identical, which pairs
the FLM one-loop determinants, generalized entropies and quantum extremal
surfaces (Theorem~\ref{thm:replicablind} and
Corollary~\ref{cor:flmblind}). Fourth, on the same filling the charged
generating functional obeys $Z_{q,\fmark}[n;\mu]=Z_{q,\fmark}[-n;-\mu]$, so
the symmetry-resolved partition functions and entropies satisfy
$\cZ_q(\cQ;n)=\cZ_q(-\cQ;-n)$ and every odd charged cumulant is exactly odd in
$n$ (Corollary~\ref{cor:chargedblind}). Fifth, in a fixed marked linked replica
filling, the normalized matter amplitude has the universal topological-limit
contrast in Eq.~\eqref{eq:central}, subject to finite-core corrections,
matched edge and framing data, and the nonaliasing condition stated in
Sec.~\ref{sec:contrast}. Sixth, that amplitude aliases the two orientations
if and only if the vortex sector is self-conjugate
(Proposition~\ref{prop:aliasing}), and away from that case the doublet is a
one-qubit code for the neutral gravitational algebra
(Proposition~\ref{prop:logical}). The interpretive statements of
Sec.~\ref{sec:noglobal} and Sec.~\ref{sec:erasure} are consequences of these
six together with results established elsewhere, and are not independent
claims of this article.

Separately, Secs.~\ref{sec:exterior} and~\ref{sec:band} derive a conditional
linearized consequence of the BPS core data: under the stationary-branch and
charge-matching assumptions, the even asymptotic shifts are those of
Eq.~\eqref{eq:exteriorcharges}; with the additional rotating-BTZ matching
diagnostic they generate the continuous envelope of
Eq.~\eqref{eq:bandcondition}. This perturbative estimate is not included among
the six exact statements above. The numerical calculation of
Sec.~\ref{sec:profiles} verifies local source parity and the BPS sum rules,
but it is not a fully nonlinear BTZ black-hole solution.

\end{document}